\documentclass[journal]{IEEEtran}
\usepackage{amsmath, amsfonts,amssymb, bm, amsthm}
\usepackage{algorithmic, algorithm, array}
\usepackage{subfig}
\usepackage{graphicx}
\usepackage{xcolor, balance}

\newtheorem{proposition}{Proposition}

\begin{document}

\title{Pulse Shaping Filter Design for Zak-OTFS}

\author{Kecheng Zhang, Weijie Yuan,\IEEEmembership{~Senior Member,~IEEE}, Yonghui Li\IEEEmembership{~Fellow,~IEEE}
    \thanks{Kecheng Zhang and Weijie Yuan are with the School of System Design and Intelligent Manufacturing and the Shenzhen Key Laboratory of Robotics and Computer Vision, Southern University of Science and Technology, Shenzhen 518055, China (e-mail: zhangkc2022@mail.sustech.edu.cn; yuanwj@sustech.edu.cn).
    }
    \thanks{Yonghui Li is with the School of Electrical and Information Engineering, The University of Sydney, Sydney, NSW 2006, Australia (e-mail: yonghui.li@sydney.edu.au)}
}






\maketitle

\begin{abstract}
    The Zak-transform-based Orthogonal Time Frequency Space (Zak-OTFS), offers a robust framework for high-mobility communications by simplifying the input-output (I/O) relation to a twisted convolution. While this structure theoretically enables accurate channel estimation by sampling the response from one pilot symbol, practical implementation is constrained by the spreading of effective channel response induced by pulse shaping filters. To address this, we first derive the I/O relationship for discrete-time oversampled Zak-OTFS, which closely approximates the continuous-time system and facilitates analysis and numerical simulation. We show that every delay-Doppler domain symbol undergoes the same effective channel response under the discrete oversampled Zak-OTFS. We then analyze the impact of window ambiguity functions, and reveal that high sidelobes lead to wide channel spreading and degrade estimation accuracy. Building on this insight, we propose a novel pulse shaping filter design that synthesizes Prolate Spheroidal Wave Functions (PSWFs) within the Isotropic Orthogonal Transform Algorithm (IOTA) framework. Numerical simulations confirm that the proposed design achieves superior channel estimation accuracy and bit error rate (BER) performance compared to conventional root-raised-cosine and rectangular windowing schemes in the high-SNR regime.
\end{abstract}


\section{Introduction}

The demand for highly reliable wireless communication in high-mobility scenarios is rapidly increasing. This requirement is a key expectation for next-generation wireless systems \cite{10969844, lu2024integrated}, which has driven the exploration of modulation schemes capable of operating effectively in doubly dispersive channels \cite{ISAC_review}. Orthogonal time frequency space (OTFS) modulation \cite{Initial_OTFS} has emerged as a promising candidate due to its ability to describe the channel and map information symbols onto the delay-Doppler (DD) domain. By operating in the DD domain, OTFS effectively exploits the full diversity and inherent characteristics of wireless channels, such as sparsity and quasi-static in high-mobility environment \cite{Initial_OTFS}. Extensive research has shown that these properties lead to enhanced communication performance. In the following subsection, we will recap the OTFS literature and summarize the achievements and limits of existing research about OTFS.

\subsection{Literature Review}

In the early stage of OTFS, the two-dimensional (2D) inverse symplectic finite Fourier transform (ISFFT) is implemented to convert the DD domain information symbols into the time-frequency (TF) domain and then transmit the TF domain symbols via the orthogonal frequency-division multiplexing (OFDM) structure \cite{Initial_OTFS}, which is called two-stage implementation multi-carrier OTFS (MC-OTFS) in \cite{Saif_Predictability}. MC-OTFS is highly compatible with OFDM wireless systems and is widely adopted in existing studies.

The following literature review focuses on three primary dimensions of MC-OTFS research: the input-output (I/O)  relationship \cite{OTFS_2018_Viterbo, OTFS_practical_pulse}, channel estimation techniques \cite{OTFS_embedded_pilot}, and performance analysis \cite{OTFS_ISAC_Giuseppe, Fair_Compare_OTFS_OFDM_Giuseppe}. To be specific, the work \cite{OTFS_2018_Viterbo} derived the widely adopted I/O relationship for MC-OTFS modulation, which was further simplified in \cite{OTFS_practical_pulse}, and analyzed the DD domain I/O relation under different pulses. A low-complexity message passing algorithm was also developed in \cite{OTFS_2018_Viterbo}, enabling OTFS systems to achieve error performance close to the ideal case and outperform OFDM under various channel conditions. In \cite{OTFS_embedded_pilot}, a DD-domain channel estimation method employing a high-power embedded pilot with sufficient guard intervals was proposed. Exploiting path separability in the DD domain, \cite{OTFS_embedded_pilot} achieved an accurate estimation by directly comparing the transmit pilot and receive pilot responses. A joint radar and communication system based on OTFS modulation was considered for vehicular scenarios in \cite{OTFS_ISAC_Giuseppe}. The system achieves near-optimal radar accuracy comparable to FMCW waveforms and supports high communication rates. The work \cite{Fair_Compare_OTFS_OFDM_Giuseppe} conducted a fair comparison between MC-OTFS and OFDM modulation schemes in terms of achievable communication rate under practical assumptions. Simulation results showed that MC-OTFS was insensitive to Doppler shifts and achieved a better communication rate. According to the results of these studies, MC-OTFS has the potential to obtain superior communication performance compared to OFDM in high-mobility environments.

Despite these advancements, MC-OTFS faces a fundamental challenge in achieving accurate channel estimation, particularly in the presence of fractional delay-Doppler indices or unresolvable paths. Existing works have extensively studied MC-OTFS channel estimation, such as \cite{OTFS_embedded_pilot, OTFS_off_grid, OTFS_off_grid_FanPingzhi} and the subsequent works inspired by them. However, these work rely on the simplified channel model by assuming an ideal pulse that satisfies bi-orthogonality \cite{Initial_OTFS} but is not physically realizable, or the paths with integer-valued delay and Doppler shifts indices. While such an assumption simplified the analytical process, it created a significant discrepancy with practical scenarios, making these methods and conclusion inapplicable to real-world systems. For instance, in an MC-OTFS system that adopts rectangular pulse shaping, there will be additional phase interference on every received symbol related to the delay and Doppler shifts of channel and grid indices of symbols \cite{OTFS_2018_Viterbo}. Then, the schemes proposed in these studies, which estimate the DD domain effective channel using a single pilot, will lead to a poor channel estimation performance. Instead of estimating the channel through the response of one pilot symbol, some works considered directly estimating the effective channel matrix together with the data symbols, such as \cite{OTFS_off_grid, OTFS_off_grid_FanPingzhi, OTFS_off_grid_ShanYaru}. The work \cite{OTFS_off_grid_ShanYaru} derived a sparse Bayesian inference-based algorithm for MC-OTFS with rectangular pulse, and \cite{OTFS_off_grid, OTFS_off_grid_FanPingzhi} focused on MC-OTFS system with ideal-pulse and \cite{OTFS_off_grid} further assumed the number of paths was known. These works require many iterations to converge, making it hard to meet the low-latency demands of real-time communication \cite{10596930}. As a result, efficiently obtaining accurate channel estimates becomes challenging, which hinders the practical implementation of MC-OTFS.

In recent studies, a new OTFS framework (referred to as Zak-OTFS \cite{Saif_Predictability, Saif_Mathematical_Foundation}) was proposed based on the Zak transform. Zak transform serves as a mathematical framework that reveals the underlying relationships of the same signal among time, frequency, and DD domains \cite{Zak_Transform_Definition}. By describing the time domain signal in DD domain through Zak transform, the procedure of a signal passing a doubly dispersive channel can be represented as a twisted convolution \cite{Saif_Mathematical_Foundation} between the channel response and the DD domain representation of the signal, meaning that every data symbol experience the same channel response. Due to the associative property of twisted convolution, applying the DD domain twisted convolution filtering before and after transmission of the signal will only change the effective channel response; each data symbol still goes through the same response. This makes it possible to estimate the channel accurately by observing the received response of one pilot symbol under the practical pulse, which is called the predictability of Zak-OTFS by \cite{Saif_Predictability}. As discussed in \cite{Saif_Predictability}, perform channel estimation through the response of just one symbol, the mean square error (MSE) between the estimated and actual channel can be lower than $-25$ dB, thereby positioning Zak-OTFS as a promising candidate for robust communication in practical high-mobility scenarios.

The channel predictability property of Zak-OTFS has already attracted attention of researchers and motivated related studies. The work \cite{Hanly_Transmitter} summarized two types of practical transmit frameworks of Zak-OTFS and analyzed the corresponding channel estimation performance. However, the receiver was assumed to be a delta impulse in their work \cite{Hanly_Transmitter}, which is not physically implementable. In their following work \cite{Hanly_OTFS_Receiver}, a complete physically realizable Zak-OTFS framework was proposed. Nevertheless, the framework focused on radar sensing purposes; how to perform robust communication in a high-mobility scenario was not considered. The recent work \cite{Saif_Interleaved_Pilot} introduced an interleaved pilot design for Zak-OTFS to overcome Doppler aliasing without increasing the time duration of signal, which provided robust communication performance when the Doppler shift was larger than the Doppler period in the DD domain.

Zak-OTFS has paved the way for simple, efficient, and accurate channel estimation in high-mobility scenarios. However, a pivotal question remains: how can we further enhance channel estimation accuracy—and consequently, symbol detection performance—given limited time-frequency resources? Prior studies have predominantly relied on applying standard DD domain filters, such as Sinc, RRC \cite{Saif_Predictability, Hanly_OTFS_Receiver, Saif_Interleaved_Pilot}, or Gaussian filters \cite{Zak_OTFS_closed_form_IO}, to the $\delta$-pulsone to analyze the channel estimation accuracy of the corresponding Zak-OTFS waveforms. To date, the potential of optimizing the pulse shaping filter to further improve the channel estimation performance has not been fully explored.

\subsection{Our Contributions}

In this paper, we consider designing the pulse shaping filter for more accurate channel estimation and better communication performance of Zak-OTFS. The key contributions of this paper are summarized as follows:
\begin{itemize}
    \item We derived the I/O relationship for the discrete oversampled Zak-OTFS system. By appending a cyclic prefix (CP) to the time-domain sequence, the resulting DD I/O relation remains consistent with the continuous Zak-OTFS framework. Specifically, the interaction is characterized by a discrete twisted convolution between the effective channel response and the data symbols, implying that all DD domain symbols experience an identical effective channel response.
    \item We investigate how the ambiguity function sidelobes of the time-frequency windows affect the spreading of the effective DD channel response. Our analysis confirms that windows with higher sidelobes ambiguity functions intensify channel spreading, which consequently leads to a deterioration in channel estimation performance.
    \item We propose a set of pulse shaping filters synthesized via the Isotropic Orthogonal Transform Algorithm (IOTA) framework using Prolate Spheroidal Wave Functions (PSWFs) \cite{slepian1961prolate}. Numerical results in the high-SNR regime confirm that this design outperforms the conventional RRC-windowed benchmark, offering reduced spectral width alongside enhanced channel estimation and BER.
\end{itemize}

The rest of this paper is organized as follows. We first review the I/O relation of Zak-OTFS and then derive the discrete-time oversampled Zak-OTFS I/O relation in Section \ref{sec:ZakOTFSModel}. We then introduce the Zak-OTFS channel estimation scheme and the motivation for designing the pulse shaping filter in Section \ref{sec:motivation}. In Section \ref{sec:solution}, we detail the properties of PSWFs and method for applying the IOTA framework to the PSWF-windowed pulse set. Numerical results are given in Section \ref{sec:numerical_results}. We finally make a conclusion in Section \ref{sec:conclusion}.

\section{Zak-OTFS System Model}\label{sec:ZakOTFSModel}

In this section, we first review the definition and input-output (I/O) relationship of Zak-OTFS based on continuous-time pulses and integration operations. Building upon these preliminaries, we derive the discrete-time oversampled Zak-OTFS I/O relationship.

\subsection{Zak-OTFS I/O Relation}

The Zak-OTFS \cite{Saif_Predictability} is first constructed in the DD domain by summing the multiplications between DD domain transmit symbols and the unfiltered basis functions, $\mathcal{Z}_{p^{\tau_l, \nu_k}}(\tau, \nu)$, i.e.,
\begin{equation}\label{sec2_eq:dd_transmit_signal}
    x_{\operatorname{dd}}(\tau, \nu) = \sum_{l=0}^{M-1} \sum_{k=0}^{N-1} X_{\operatorname{dd}}[l, k] \mathcal{Z}_{p^{\tau_l, \nu_k}}(\tau, \nu),
\end{equation}
where $X_{\operatorname{dd}}[l, k]$ is the information symbols at DD domain grid $(l, k) \in \mathcal{Q}$, $\tau_l=\frac{l}{M \Delta f}$, $\nu_k=\frac{k}{NT}$, $T$ and $\Delta f$ satisfying $T \Delta f = 1$ are the delay and Doppler periods, respectively and the unfiltered basis function is defined as
\begin{multline}\label{sec2_eq:dd_pulse}
    \mathcal{Z}_{p^{\tau_l, \nu_k}}(\tau, \nu)   =  \sum_{m \in \mathbb{Z}} \sum_{n \in \mathbb{Z}}                                                                                                             \\
    e^{j 2 \pi \nu_k n T} \delta\left(\tau-\tau_l-n T\right) \delta\left(\nu-\nu_k-m \Delta f\right).
\end{multline}
By implementing the inverse Zak-transform \cite{Saif_Predictability} on $x_{\operatorname{dd}}(\tau, \nu)$, the corresponding time domain transmit signal is given by
\begin{equation}\label{sec2_eq:td_transmit_signal}
    \begin{aligned}
        x(t) & = \sqrt{T} \int_{0}^{\Delta f} x_{\operatorname{dd}}(t, \nu) \,d\nu                   \\
             & = \sum_{l=0}^{M-1} \sum_{k=0}^{N-1} X_{\operatorname{dd}}[l, k] p^{\tau_l, \nu_k}(t),
    \end{aligned}
\end{equation}
where $p^{\tau_l, \nu_k}(t)$ is the corresponding time domain unfiltered basis signal after performing Zak-transform,
\begin{equation}\label{sec2_eq:delta_pulsone}
    p^{\tau_l, \nu_k}(t) = \sqrt{T} \sum_{n \in \mathbb{Z}} e^{j 2 \pi \nu_k n T} \delta(t - \tau_l - nT).
\end{equation}
Since the signal \eqref{sec2_eq:td_transmit_signal} occupies infinite time and frequency domain resources, it is not physically realizable. Both the time and frequency domain resources need to be restricted to implement the Zak-OTFS in practice.

According to \cite{Saif_Predictability, Hanly_Transmitter}, we can limit the occupied time and frequency resources by performing a twisted convolution\footnote{Consider $a(\tau, \nu)$ and $b(\tau, \nu)$ are function defined on $\mathbb{R}^{2}$, the twisted convolution between them is defined as $a *_{\sigma} b (\tau, \nu) = \iint_{\mathbb{R}^{2}} a(\tau^{\prime}, \nu^{\prime})b(\tau - \tau^{\prime}, \nu - \nu^{\prime}) e^{j 2 \pi \nu^{\prime} (\tau - \tau^{\prime})} \, d\tau^{\prime} \, d \nu^{\prime}$.} between a DD domain transmit filter, $g(\tau, \nu)$, and the DD domain transmit signal in \eqref{sec2_eq:dd_transmit_signal}. The DD domain twisted convolution filtered signal is given by
\begin{equation}\label{sec2_eq:dd_filtered_transmit_signal}
    \mathcal{Z}_{x_{g}}(\tau, \nu) = g *_{\sigma} \mathcal{Z}_{x}(\tau, \nu),
\end{equation}
where the corresponding time domain signal is denoted as $x_{g}(t)$. We obtain the filtered time domain transmit signal by implementing the inverse Zak-transform on \eqref{sec2_eq:dd_filtered_transmit_signal}. By denoting the doubly-dispersive channel as $h(\tau, \nu)$, the time domain receive signal is given by
\begin{equation}\label{sec2_eq:td_receive_signal}
    r(t)=\iint h(\tau, \nu) x_{g}(t-\tau) e^{j 2 \pi \nu(t-\tau)} d \tau d \nu + n(t),
\end{equation}
where $n(t)$ is the time domain noise. Now we implement Zak-transform \cite{Saif_Predictability} on \eqref{sec2_eq:td_receive_signal}, the corresponding DD domain receive signal can be represented as \cite{Saif_Predictability,Hanly_OTFS_Receiver}
\begin{multline}\label{sec2_eq:dd_receive_signal}
    \mathcal{Z}_{r}(\tau, \nu) = \sqrt{T} \sum_{k=-\infty}^{\infty} r(\tau + kT) e^{-j 2 \pi k \nu T}\\
    = h *_\sigma \mathcal{Z}_{x_g}(\tau, \nu) = h *_\sigma\left(g *_\sigma x_{\operatorname{dd}}\right)(\tau, \nu).
\end{multline}
The receiver performs another twisted convolution between a DD domain receiving filter, $\tilde{g}(\tau, \nu)$, and the DD domain receive signal \eqref{sec2_eq:dd_receive_signal}, and gets the DD domain filtered receive signal,
\begin{equation}
    y_{\operatorname{dd}}(\tau, \nu) = \tilde{g} *_{\sigma} \mathcal{Z}_{r}(\tau, \nu).
\end{equation}
Combining the results above, the DD domain I/O relation of Zak-OTFS is given by
\begin{equation}
    y_{\operatorname{dd}}(\tau, \nu) = h_{\operatorname{eff}}(\tau, \nu) *_{\sigma} x_{\operatorname{dd}}(\tau, \nu) + \tilde{n}_{\operatorname{dd}}(\tau, \nu),
\end{equation}
where $h_{\operatorname{eff}}(\tau, \nu) = \tilde{g} *_{\sigma} h *_{\sigma} g(\tau, \nu)$ is the effective channel response, and $\tilde{n}_{\operatorname{dd}}(\tau, \nu) = \tilde{g} *_{\sigma} \mathcal{Z}_{n}(\tau, \nu)$ is the DD domain effective noise. By sampling $y_{\operatorname{dd}}(\tau, \nu)$ at $\tau = \tau_l$ and $\nu = \nu_k$ and omitting the noise term, we have the discrete twisted convolution \cite{Saif_Predictability} in DD domain,
\begin{equation}\label{sec2_eq:DD_discret_IO}
    Y_{\operatorname{dd}}[l, k] = \sum_{l^{\prime} \in \mathbb{Z}} \sum_{k^{\prime} \in \mathbb{Z}} h_{\operatorname{eff}}[l^{\prime}, k^{\prime}] X_{\operatorname{dd}}[l-l^{\prime}, k-k^{\prime}] e^{j 2 \pi \frac{k^{\prime} (l - l^{\prime})}{MN}},
\end{equation}
where $Y_{\operatorname{dd}}[l, k]$ is the DD domain received symbols at $(l, k)$, $h_{\operatorname{eff}}[l, k]$ and $X_{\operatorname{dd}}[l, k]$ refer to sampling $h_{\operatorname{eff}}(\tau, \nu)$ and $x_{\operatorname{dd}}(\tau, \nu)$ at $\tau = \tau_l$ and $\nu = \nu_k$, respectively.

Since the samples $y_{\operatorname{dd}}[l, k]$ are quasi-periodic, $MN$ samples of $y_{\operatorname{dd}}[l, k]$ with $(l, k) \in \mathcal{M}_{M}\times\mathcal{M}_{N}$ contain all the information of $y_{\operatorname{dd}}[l, k]$ for any $(l, k) \in \mathbb{Z}^2$. Summarizing the results above, the DD domain input-output relation can be rewritten in matrix-form as
\begin{equation}\label{Sec2_input_output_matrix}
    \mathbf{y} = \mathbf{H} \mathbf{x} + \mathbf{n},
\end{equation}
where the $(kM + l)$-th elements in $\mathbf{x} \in \mathbb{C}^{MN \times 1}$, $\mathbf{y} \in \mathbb{C}^{MN \times 1}$, and $\mathbf{n} \in \mathbb{C}^{MN \times 1}$ are the samples of $x_{\operatorname{dd}}(\tau, \nu)$, $y_{\operatorname{dd}}(\tau, \nu)$, and $\tilde{n}_{\operatorname{dd}}(\tau, \nu)$ at $\tau = \tau_l$ and $\nu = \nu_k$, respectively, and the $(k^{\prime}M + l^{\prime}, kM + l)$-th element in $\mathbf{H} \in \mathbb{C}^{MN \times MN}$ is given by \cite{Saif_Predictability}
\begin{multline}\label{Sec2_effective_channel_ele}
    \hspace{-1em}H[k^{\prime}M + l^{\prime}, kM + l] = \sum_{(m, n) \in \mathbb{Z} \times \mathbb{Z}} h_{\operatorname{eff}}[l^{\prime} - l - nM, k^{\prime} - k - mN] \\
    e^{j 2 \pi \frac{nk}{N}} e^{j 2 \pi \frac{(k^{\prime} - k - mN)(l + nM)}{MN}},
\end{multline}

\subsection{Specific Zak-OTFS Implementation via Time and Frequency Windowing}\label{sec2_2_DD_IO_OFDM_based}
\begin{figure*}[t]
    \centering
    \includegraphics[width=1\textwidth]{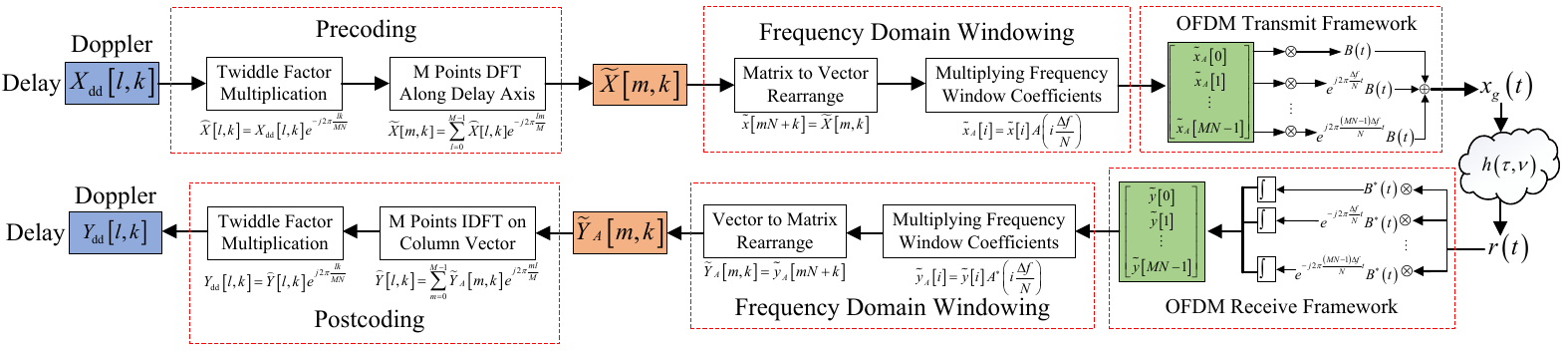}
    \caption{The transmit and receive framework of Zak-OTFS based on OFDM \cite{Hanly_Transmitter, Zak_OTFS_Receiver}.}\label{sec2_fig_OFDM_transceiver}
\end{figure*}
In this subsection, we introduce a specific implementation of Zak-OTFS that windowing the $\delta$-pulsone \eqref{sec2_eq:delta_pulsone} first in frequency domain and then in time domain. This specific type of Zak-OTFS can be expressed by a twisted convolution between the DD domain transmit filter and \eqref{sec2_eq:dd_pulse} is constructed as
\begin{equation}\label{sec2_eq:dd_filter}
    g(\tau, \nu) = \alpha(\tau) \beta(\nu) e^{j 2 \pi \tau \nu}.
\end{equation}
According to \cite{Hanly_Transmitter}, performing the twisted convolution of \eqref{sec2_eq:dd_filtered_transmit_signal} with $g(\tau, \nu)$ given in \eqref{sec2_eq:dd_filter} is equivalent to implement windowing on \eqref{sec2_eq:td_transmit_signal} first in frequency domain and then in time domain. The windowed time domain basis function is expressed as
\begin{multline}\label{Sec2_windwos_pulse}
    p^{\tau_l, \nu_k}_{g}(t) = B(t) e^{j 2 \pi \nu_k (t - \tau_l)} \\
    \times \sum_{m \in \mathbb{Z}} A(m \Delta f + \nu_k) e^{j 2 \pi m \Delta f (t - \tau_l)},
\end{multline}
where $A(f) = \int_{\mathbb{R}} \alpha(\tau) e^{-j 2 \pi f \tau} \, d\tau$ is the frequency domain window function, $B(t) = \int_{\mathbb{R}} \beta(\nu) e^{j 2 \pi \nu t} \, d\nu$ is the time domain window function, $p^{\tau_l, \nu_k}_{g}(t)$ is known as \textit{pulsone} \cite{Saif_Predictability}.

We assume the channel is sparse and there are $Q$ resolvable paths, which gives
\begin{equation}
    h(\tau, \nu) = \sum_{q=1}^{Q} h_{q} \delta(\tau - \tau_{q}) \delta(\nu - \nu_{q}),
\end{equation}
where $\tau_q$ and $\nu_q$ are the delay and Doppler shift of the $q$-th path, respectively. Then time-domain received signal after passing the channel then becomes
\begin{equation}
    r(t) = \sum_{q=1}^{Q} h_{q} x_{g}(t - \tau_{q}) e^{j 2 \pi \nu_{q} (t - \tau_{q})} + n(t).
\end{equation}
The DD domain receive filter is constructed as
\begin{equation}\label{sec2_eq:dd_receive_filter}
    \tilde{g}(\tau, \nu) = \alpha^{*}(-\tau) \beta^{*}(-\nu).
\end{equation}
Performing the twisted convolution between $\tilde{g}(\tau, \nu)$ given in \eqref{sec2_eq:dd_receive_filter} and $r(t)$ is equivalent to time windowing the received signal $r(t)$ with $B^{*}(t)$ followed by frequency windowing it with $A^{*}(f)$. Then, the effective channel response can be expressed as \cite{Hanly_OTFS_Receiver}
\begin{multline}\label{Sec2_effective_channel_2}
    h_{\operatorname{eff}}(\tau, \nu) = \sum_{q=0}^{Q-1} h_q e^{j 2 \pi (\tau \nu - \tau_q \nu_q)} \\
    \times \mathcal{Y}_{A}(\tau - \tau_q, -\nu) \mathcal{X}_{B}(-\tau_p, \nu - \nu_{q}),
\end{multline}
where $\mathcal{Y}_{A}(\tau, \nu) = \int_{\mathbb{R}} A(f) A^{*}(f - \nu) e^{j 2 \pi f \tau} \, df$ and $\mathcal{X}_{B}(\tau, \nu) = \int_{\mathbb{R}} B(t) B^{*}(t - \tau) e^{-j 2 \pi \nu (t-\tau)} \, dt$ are the ambiguity functions of $A(f)$ and $B(t)$, respectively.

According to the results in \cite{Hanly_OTFS_Receiver}, the noise process after twisted convolution filtering is a zero mean Gaussian process with covariance
\begin{multline}
    C_{\mathcal{Z}_{\tilde{n}}}(\tau, \nu \mid \tau^{\prime}, \nu^{\prime}) = \mathbb{E}[\mathcal{Z}_{\tilde{n}}(\tau, \nu) \mathcal{Z}_{\tilde{n}}(\tau^{\prime}, \nu^{\prime})]\\
    = \frac{N_0}{T} \tilde{g} *_{\sigma} g *_{\sigma} \mathcal{Z}_{p^{\tau^{\prime}, \nu^{\prime}}}(\tau, \nu),
\end{multline}
where $n(t)$ is a white Gaussian noise process with power spectral density $N_0$ Watts/Hz, and $\mathcal{Z}_{p^{\tau^{\prime}, \nu^{\prime}}}(\tau, \nu) = \sum_{(m, n) \in \mathbb{Z} \times \mathbb{Z}} e^{j 2 \pi \nu^{\prime} n T} \delta(\tau - \tau^{\prime} - nT)\delta(\nu - \nu^{\prime} - m \Delta f)$. When the window functions $A(f)$ and $B(t)$ are root Nyquist window \cite{Hanly_OTFS_Receiver}, the samples of $C_{\mathcal{Z}_{\tilde{n}}}(\tau, \nu \mid \tau^{\prime}, \nu^{\prime})$ at $(\tau_l, \nu_k, \tau_{l^{\prime}}, \nu_{k^{\prime}})$ are uncorrelated, i.e., the twisted convolution filtered noise still follows the white Gaussian process \cite{Hanly_OTFS_Receiver}. Otherwise, the $(kM+l, k^{\prime}M + l^{\prime})$-th element of noise samples correlation matrix $\mathbf{C}_{n} = \mathbb{E}[\mathbf{n} \mathbf{n}^{\operatorname{T}}]$ is given by
\begin{multline}\label{Sec2_noise_covariance_ele}
    C_{\mathcal{Z}_{\tilde{n}}}(\tau_l, \nu_k \mid \tau_{l^{\prime}}, \nu_{k^{\prime}}) = \frac{1}{T} \sum_{(m, n) \in \mathbb{Z} \times \mathbb{Z}} \mathcal{Y}_{A}(\tau_{l} - \tau_{l^{\prime}} - nT, 0) \\
    B(\tau_{l} - (n+m)T) B(\tau_{l^{\prime}} - mT) e^{j 2 \pi \nu_k nT} e^{j 2 \pi (\nu_k - \nu_k^{\prime})mT}.
\end{multline}

\subsection{Zak-OTFS Implementation through OFDM Framework}

Fig. \ref{sec2_fig_OFDM_transceiver} shows the transmit and receiver framework of Zak-OTFS based on OFDM, and in this subsection, we briefly verify that this architecture constitutes a valid implementation of Zak-OTFS. It has been proved in \cite{Hanly_Transmitter} that the transmit framework is equivalent to the modulation procedure \eqref{sec2_eq:td_transmit_signal}. When receiving $Y_{\operatorname{dd}}[l, k]$, the time domain received signal is placed on an OFDM receiver, and the $i$-th frequency receive response is given by
\begin{equation}
    \tilde{y}[i] = \int r(t) B^{*}(t) e^{-j 2 \pi \frac{i \Delta f}{N}t} \, dt.
\end{equation}
Then, the receiver performs frequency domain windowing and postcoding with $M$ points inverse discrete Fourier transform (IDFT) and twiddle factor multiplication, which can be expressed as
\begin{multline}\label{sec2_OFDM_receiver}
    Y_{\operatorname{dd}}[l, k] = \int r(t) \sum_{m \in \mathbb{Z}} B^{*}(t) e^{-j 2 \pi \frac{(m \Delta f + \nu_{k}) \Delta f}{N}t} \\
    \times A^{*}(m \Delta f + \nu_{k}) e^{j 2 \pi \frac{ml}{M}} e^{j 2 \pi \frac{lk}{MN}} \, dt.
\end{multline}
According to the results in \cite{Hanly_OTFS_Receiver}, the receiving procedure of Zak-OTFS is equivalent to the correlation receiver, i.e.,
\begin{equation}\label{Sec2_correlation_receive}
    Y_{\operatorname{dd}}[l, k] = T \int r(t) \left[p^{\tau_l, \nu_k}_{g}(t)\right]^{*} \,dt.
\end{equation}
Through straightforward algebraic manipulation, it becomes evident that \eqref{sec2_OFDM_receiver} is equivalent to \eqref{Sec2_correlation_receive}. This confirms that the reception process depicted in Fig. \ref{sec2_fig_OFDM_transceiver} corresponds precisely to the Zak-OTFS reception scheme.

\subsection{Discrete-Time Oversampled Zak-OTFS I/O Relation}\label{sec2_3_discrete_DD_IO_based}

Transitioning from continuous-domain theory to a discrete-time oversampled I/O relationship is essential for practical system realization. This derivation bridges the gap between theory and digital implementation, providing the necessary framework for Zak-OTFS simulations and hardware deployment.

We assume the sampling rate for the pulse \eqref{Sec2_windwos_pulse} is $F_s = L M \Delta f$ with $L \geq 1$ being the oversampling factor. Then, the $n$-th sample of the pulse $p^{\tau_l, \nu_k}_{g}(t)$ is given by
\begin{equation}
    p^{\tau_l, \nu_k}_{g}[n] = B[n] \sum_{m \in \mathbb{Z}} A(m \Delta f + \nu_k) e^{j 2 \pi (m \Delta f + \nu_k) (n - \tau_l)}.
\end{equation}
Then, the time domain transmit samples are represented as
\begin{equation}\label{Sec2_discrete_transmit_TD_symbols}
    s[n] = \sum_{(l, k) \in \mathcal{Q}} X_{\operatorname{dd}}[l, k] p^{\tau_l, \nu_k}_{g}[n].
\end{equation}
By denoting $\mathbf{x}_{\operatorname{dd}} = \operatorname{vec}(\mathbf{X}_{\operatorname{dd}})$, and building a matrix $\mathbf{P}$ with each column being the samples of the pulse, i.e.,
\begin{equation}\label{Sec2_sampled_pulse_matrix}
    \mathbf{P} = \left[
        \begin{matrix}
            \mathbf{p}^{\tau_0, \nu_0}_{g}, \mathbf{p}^{\tau_1, \nu_0}_{g}, \dots, \mathbf{p}^{\tau_{M-1}, \nu_{N-1}}_{g}
        \end{matrix}
        \right],
\end{equation}
where $\mathbf{p}^{\tau_l, \nu_k}_{g} \in \mathbb{C}^{N F_s \times 1}$ denotes the vector of sampled pulses of the windowed pulses on information grid $(\tau_l, \nu_k)$ with sample rate $F_s$, we obtain the vectorized form of \eqref{Sec2_discrete_transmit_TD_symbols},
\begin{equation}\label{Sec2_TD_symbols_vector}
    \mathbf{s} = \mathbf{P} \mathbf{x}_{\operatorname{dd}}.
\end{equation}

Without loss of generality, we assume a sufficiently high oversampling rate $L$ to ensure that all propagation delays align precisely with the discrete sampling grid so that the delay of the $q$-th path is characterized by an integer multiplication of sampling interval $T_s = 1 / F_s$. Assume that the cyclic prefix (CP) length exceeds the maximum delay across all $Q$ resolvable paths. After discarding the CP at the sensing receiver, the resulting discrete-time echo signal is then expressed as:
\begin{equation}\label{Sec2_TD_receive_symbols_vector}
    \mathbf{r} = \sum_{q=1}^{Q} h_{q} \mathbf{J}_{\tau_{q}} \mathbf{D}_{\nu_{q}} \mathbf{s} + \mathbf{n},
\end{equation}
where $\mathbf{n} \sim \mathcal{CN}(0, \sigma_{n}^{2} \mathbf{I}_{MN})$ stands for the white Gaussian noise, the matrix $\mathbf{J}_{k} \in \mathbb{C}^{LMN \times LMN}$ represents the $k$-th cyclic shift operator modeling the CP-induced circularity,
\begin{equation}
    \mathbf{J}_{k} = \left[
        \begin{matrix}
            \mathbf{0}           & \mathbf{I}_{k} \\
            \mathbf{I}_{LMN - k} & \mathbf{0}
        \end{matrix}
        \right], \mathbf{J}_{-k} = \mathbf{J}^{\operatorname{T}},
\end{equation}
and
\begin{equation}
    \mathbf{D}_{\nu} = \operatorname{Diag}\left\{\left[1, e^{j2\pi\frac{\nu}{LMN}}, \dots, e^{j2\pi\frac{(LMN-1)\nu}{LMN}}\right]\right\}.
\end{equation}

We now approximate the integral \eqref{Sec2_correlation_receive} with cross-correlation,
\begin{equation}\label{Sec2_discrete_correlation_receive}
    y_{\operatorname{dd}}[l, k] = \frac{T}{F_s} \sum_{n = 0}^{LMN-1} r[n] \left(p^{\tau_l, \nu_k}_{g}[n]\right)^{*},
\end{equation}
where $r[n]$ is the $n$-th entry of $\mathbf{r}$. Combine \eqref{Sec2_TD_symbols_vector}, \eqref{Sec2_TD_receive_symbols_vector}, and \eqref{Sec2_discrete_correlation_receive} together, the DD domain I/O relation is given by,
\begin{equation}
    \mathbf{y}_{\operatorname{dd}} = \mathbf{H}_{\operatorname{eff}} \mathbf{x}_{\operatorname{dd}} + \mathbf{P}^{\operatorname{H}} \mathbf{n},
\end{equation}
where $\mathbf{H}_{\operatorname{eff}}$ denotes the DD domain effective channel matrix,
\begin{equation}\label{Sec2_discrete_H_matrix}
    \mathbf{H}_{\operatorname{eff}} = \sum_{q=1}^{Q} h_{q} \mathbf{P}^{\operatorname{H}} \mathbf{J}_{\tau_{q}} \mathbf{D}_{\nu_{q}} \mathbf{P}.
\end{equation}
We then have the following proposition about the DD domain I/O relation of the discrete-time oversampled Zak-OTFS:
\begin{proposition}\label{Sec2_proposition_effective_channel_expression}
    The $(k^{\prime}M + l^{\prime}, kM + l)$-th element in the effective channel matrix of \eqref{Sec2_discrete_H_matrix} can also be represented in \eqref{Sec2_effective_channel_ele}, where $h_{\operatorname{eff}}[l, k]$ is the DD domain effective channel response given by
    \begin{multline}
        h_{\operatorname{eff}}[l, k] = \sum_{q=1}^{Q} h_{q} e^{j 2 \pi \frac{l k - l_{\tau_{q}} k_{\nu_{q}}}{MN}}\\
        \times \mathcal{Y}_{A}(\tau_{l} - \tau_{q}, -\nu_{k}) \mathcal{X}_{\tilde{B}}(-l_{\tau_{q}}, k - k_{\nu_{q}}),
    \end{multline}
    with $l_{\tau_{q}} = \tau_{q} M \Delta f$, $k_{\nu_{q}} = \nu_{q} N T$, and $\mathcal{X}_{\tilde{B}}(l, k)$ being the discrete periodic ambiguity function \cite{1057703} of the samples of time window function $\{B(n / F_s)\}_{n=0}^{LMN-1}$ defined as
    \begin{equation}\label{Sec2_discrete_periodic_AF_B}
        \mathcal{X}_{\tilde{B}}(l, k) = \sum_{n = 0}^{LMN-1} B[n] B[(n - l)_{LMN}] e^{-j 2 \pi \frac{k (n-l)}{LMN}}.
    \end{equation}
\end{proposition}
\begin{proof}
    See Appendix \ref{apdx_effective_channel}.
\end{proof}
Proposition \ref{Sec2_proposition_effective_channel_expression} indicates that every DD domain symbol undergoes the same effective channel response under the discrete-time oversampled Zak-OTFS framework, and the spreading of effective channel response is determined by the ambiguity function of frequency window and the discrete periodic ambiguity function of time window samples.

\section{Channel Estimation with Embedded Pilot Scheme}\label{sec:motivation}

In this section, we first recap how to estimate the channel with embedded pilot scheme in Section \ref{Sec3_Crystalline_Regime}. We then show that the window function with high sidelobes in the ambiguity function leads to energy leakage, thereby motivating the need for sidelobe suppression to enhance channel estimation accuracy in Section \ref{Sec3_Sidelobe_Suppression}.

\begin{figure}[t]
    \centering
    \includegraphics[width=0.9\columnwidth]{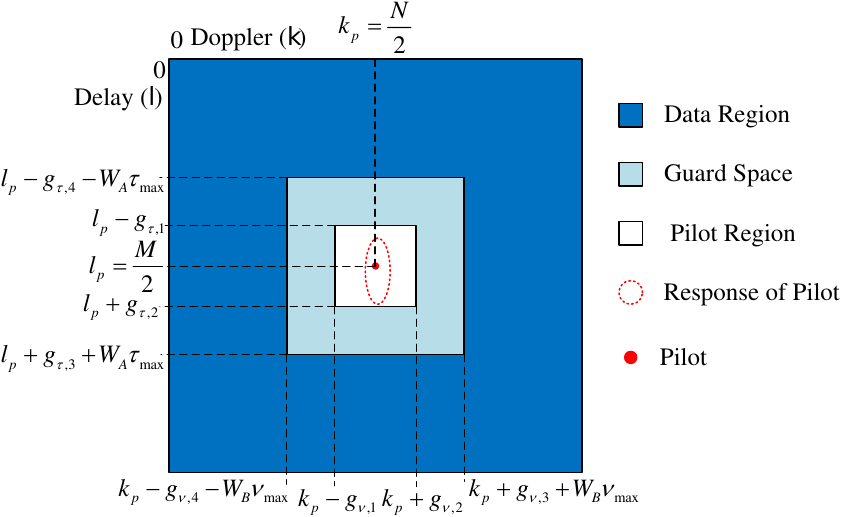}
    \caption{An OTFS frame with embedded pilot, guard space, and data.}\label{Sec3_fig_embedded_pilot_illustration}
\end{figure}
\subsection{Crystalline Regime Condition and the Embedded Pilot Channel Estimation Scheme}\label{Sec3_Crystalline_Regime}

Assuming the pilot, denoted as $x_{p}$, is placed at $(l_p, k_p) = \left(\lfloor \frac{M}{2} \rfloor, \lfloor \frac{N}{2} \rfloor\right)$, the received response from \eqref{sec2_eq:DD_discret_IO} is represented as
\begin{multline}\label{Sec3_pilot_response}
    Y_{\operatorname{dd}, p}[l, k] =  x_p h_{\operatorname{eff}}[l - l_p, k - k_p] e^{j 2 \pi \frac{(k - k_p)l_p}{MN}} \\
    + \sum_{(m, n) \in \mathbb{Z} \times \mathbb{Z} \backslash \{(0, 0)\}} x_p h_{\operatorname{eff}}[l - l_p - nM, k - k_p - mN] \\
    \times e^{j 2 \pi \frac{nk}{N}} e^{j 2 \pi \frac{(k - k_p - mN)(l_p + nM)}{MN}}.
\end{multline}
Under the crystalline regime condition\footnote{For the time and frequency windowed Zak-OTFS signal model in \eqref{Sec2_input_output_matrix}, the crystalline regime condition \cite{Saif_Predictability} holds if $\tau_{\max} < T$, $\nu_{\max} < \Delta f$, $W_A \gg \Delta f$, and $W_B \gg T$, where $\tau_{\max}$ and $\nu_{\max}$ are the maximum delay and Doppler shift of the propagation paths, respectively.}, the second term in \eqref{Sec3_pilot_response} can be ignored so that the effective channel response can be approximated by the $MN$ samples in the primal crystalline regime, i.e., the rectangular area $(\tau, \nu) \in [-T/2, T/2] \times [-\Delta f/2, \Delta f/2]$. Therefore, the effective channel response can be estimated by
\begin{equation}\label{Sec3_channel_estimation_scheme}
    \hat{h}_{\operatorname{eff}}[l, k] = \begin{cases}
        \frac{y_{\operatorname{dd}}[l+l_p, k+k_p]}{x_p} e^{-j 2\pi \frac{k l_p}{MN}},                     \\
        \;\;\;\;\operatorname{ if } -\frac{M}{2} \leq l < \frac{M}{2}, -\frac{N}{2} \leq k < \frac{N}{2}, \\
        0,                                        \operatorname{ otherwise}.
    \end{cases}
\end{equation}

For the sake of simplicity, we consider estimating the effective channel through the embedded pilot scheme. The corresponding frame structure is shown in Fig. \ref{Sec3_fig_embedded_pilot_illustration}, where $\{g_{\tau, i}\}_{i=1}^{4}$ and $\{g_{\nu, i}\}_{i=1}^{4}$ are nonnegative integers determined by the duration and bandwidth of the OTFS frame. The effective channel response is obtained from the received responses in the `Pilot Region'.

\subsection{Motivation to Suppress the Sidelobes of the Ambiguity Function}\label{Sec3_Sidelobe_Suppression}

\begin{figure}[t]
    \centering
    \subfloat[Normalized ambiguity function of $A(f)$ being a rectangular function.]{
        \includegraphics[width=0.46\columnwidth]{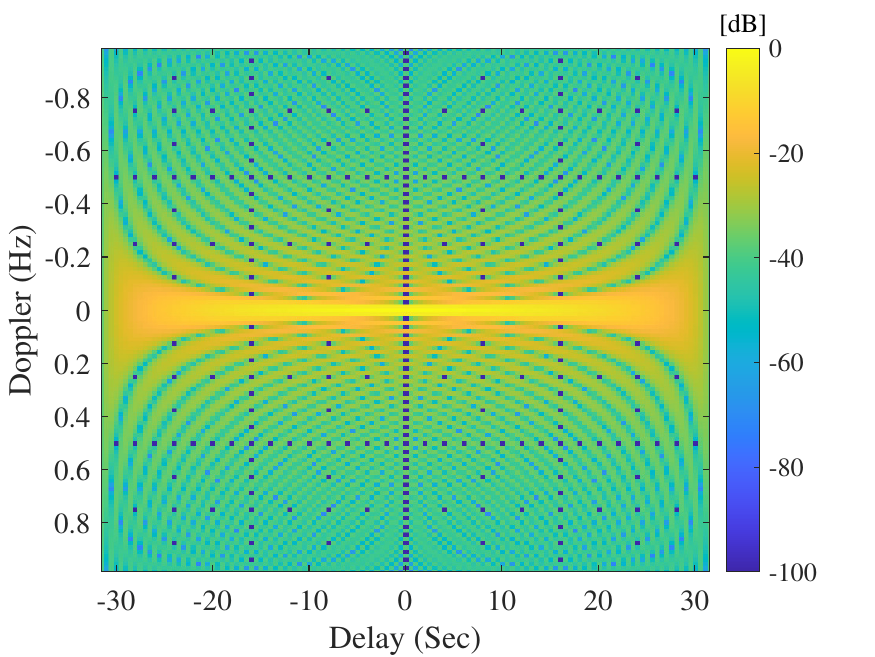}
        \label{Sec3_fig_rect_af}
    }
    \hspace{0.25em}
    \subfloat[Normalized ambiguity function of $A(f)$ with $\alpha(\tau)$ being an RRC function.]{
        \includegraphics[width=0.46\columnwidth]{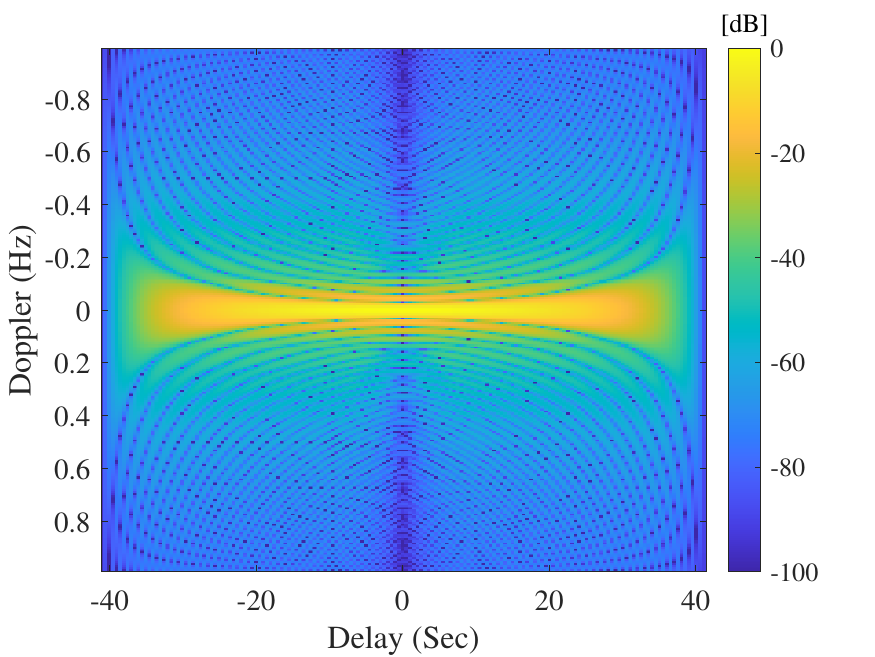}
        \label{Sec3_fig_rrc_af}
    }
    \caption{Ambiguity functions under different window functions.}
    \label{Sec3_fig_afs}
\end{figure}

Fig. \ref{Sec3_fig_embedded_pilot_illustration} shows that to achieve accurate channel estimation under the embedded pilot scheme, it requires the effective channel response to be localized in the `Pilot Region' as much as possible. According to \eqref{Sec2_effective_channel_2}, the ambiguity functions of the time and frequency windows jointly determine the spreading extent of the effective channel, and the window with a low sidelobe level ambiguity function will provide a more localized effective channel response. To show this illustratively, we provide a toy example here. Taking $M=32$, $N=32$, $T=1$, and $\Delta f = 1$ for simplicity, and assuming the channel contains only a single-path with $h(\tau, \nu) = \delta(\tau - 4 \frac{T}{M}) \delta(\nu - 4.7 \frac{\Delta f}{N})$. Fig. \ref{Sec3_fig_afs} shows that the ambiguity function of the RRC window has a lower sidelobe level than the rectangular window. As shown in Fig. \ref{Sec3_fig_effective_channel_response}, the effective channel response of Zak-OTFS under the RRC windows is more localized compared to that under the rectangular window. This suggests that by suppressing the sidelobes of the ambiguity function associated with the window function under a limited time-frequency resource, it achieves a more localized effective channel response, potentially leading to improved channel estimation accuracy, thereby enhancing the communication performance.

\begin{figure}[t]
    \centering
    \subfloat[Normalized effective channel response with $\alpha(\tau)$ and $\beta(\nu)$ being Sinc functions.]{
        \includegraphics[width=0.46\columnwidth]{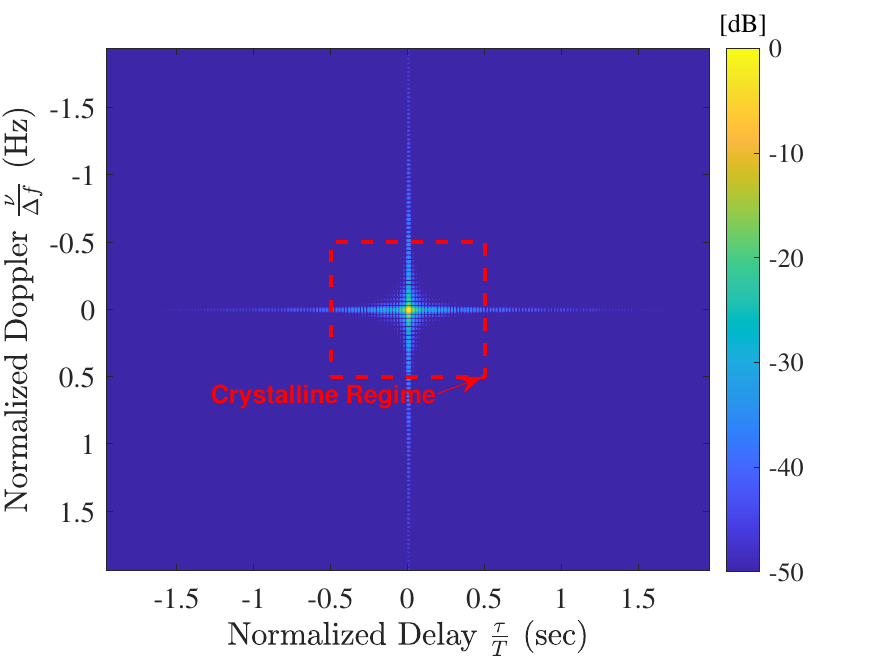}
        \label{Sec3_fig_effective_channel_response_rectangle}
    }
    \hspace{0.25em}
    \subfloat[Normalized effective channel response with $\alpha(\tau)$ and $\beta(\nu)$ being RRC functions, roll-off factor $0.3$.]{
        \includegraphics[width=0.46\columnwidth]{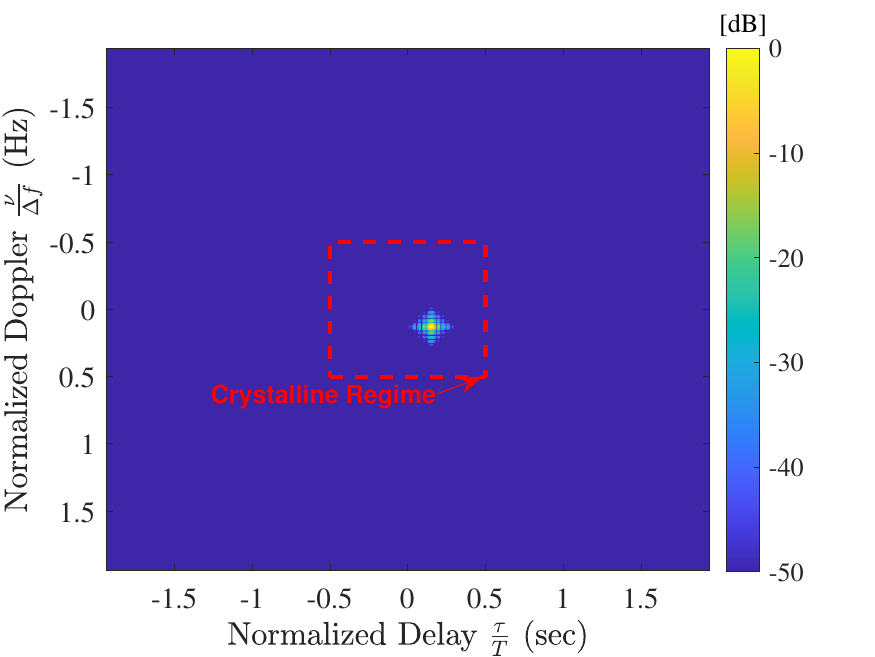}
        \label{Sec3_fig_effective_channel_response_hamming}
    }
    \caption{Effective channel responses under different window functions.}
    \label{Sec3_fig_effective_channel_response}
\end{figure}

\section{Proposed IOTA Pulse Shaping Design Framework}\label{sec:solution}

In this section, we design the pulse shaping filter through the IOTA to obtain a more localized effective channel response.

Motivated by the conclusion in Section \ref{Sec3_Sidelobe_Suppression} that low-sidelobe ambiguity functions yield the desired localized channel responses, our goal is to optimize the window design. However, globally suppressing sidelobes on the entire delay-Doppler plane is a well-known intractable problem in radar signal processing \cite{1057703}. To overcome this, we employ the IOTA framework, which simplifies this complex 2D optimization into a tractable autocorrelation function sidelobe suppression task. From the result in \cite[Lemma 1]{Hanly_OTFS_Receiver}, we have
\begin{equation}\label{Sec2_approx_AF}
    \begin{aligned}
        \mathcal{Y}_{A}(\tau, \nu) & = \mathcal{Y}_{A}(\tau, 0) + \mathcal{O}\left(\sqrt{\frac{|\nu|}{W_A}}\right), \\
        \mathcal{X}_{B}(\tau, \nu) & = \mathcal{X}_{B}(0, \nu) + \mathcal{O}\left(\sqrt{\frac{|\tau|}{W_B}}\right).
    \end{aligned}
\end{equation}
Recap the effective channel response in \eqref{Sec2_effective_channel_2}, we see from \eqref{Sec2_approx_AF} that suppressing the sidelobes of $\mathcal{Y}_{A}(\tau, 0)$, i.e., the autocorrelation function of $\alpha(\tau)$, and $\mathcal{X}_{B}(0, \nu)$, i.e., the autocorrelation function of $\beta(\nu)$, will lead to a more localized effective channel response.

Now the problem reduces to maximizing energy concentration in time (frequency) domain subject to a finite support constraint in the Doppler (delay) domain. This specific optimization challenge is uniquely solved by PSWFs \cite{slepian1961prolate}, also known as Slepian functions. PSWFs are a family of band-limited functions that maximize the energy fraction contained within a specified time interval, and they are orthogonal to each other. For any given time-bandwidth product, the zeroth-order PSWF represents the theoretically optimal window function for simultaneous time-frequency localization. The expressions of 1-D PSWFs in the $n$-th order is given by the eigenvectors of the Fredholm integral:
\begin{equation}
    \int_{-T^{\prime}}^{T^{\prime}} \operatorname{sinc}(B^{\prime} (t - t^{\prime})) \psi_{n}(t^{\prime}) = \lambda_{n} \psi_{n}(t),
\end{equation}
where the eigenvalue $\lambda_{n}$ represent the fraction of energy contained within the bandwidth $B^{\prime}$, and the corresponding eigenfunction $\psi_{0}(t)$ of the first eigenvalue $\lambda_{0}$ represented the maximally localized pulse. We construct the frequency window function band-limited in frequency domain interval $[-M \Delta f / 2, M \Delta f / 2]$ and maximally localized within time domain interval $[-1 / (2 M \Delta f), 1 / (2 M \Delta f)]$. Similarly, we construct the time window function time-limited in time domain interval $[- NT/2, NT/2]$ and maximally localized within the frequency domain interval $[-1/(2NT), 1/(2NT)]$.

Although applying PSWF windowing to the $\delta$-pulsone yields the most localized effective channel response, it does not provide the optimal symbol detection performance. In \eqref{Sec2_noise_covariance_ele}, since the PSWF windows are not root-Nyquist, it shows that the effective noise in the DD domain will be colored Gaussian noise rather than white if the time and frequency windows are not root-Nyquist. Furthermore, the non-root-Nyquist nature induces intrinsic interference between the symbols in DD domain.

One straightforward way to address this non-root-Nyquist issue is to apply IOTA on the DD domain filter. For example, in \cite{mehrotra2025pulse}, the authors utilizes the IOTA to synthesize DD pulse shaping filters by orthogonalizing the PSWFs, with respect to shifts on the information lattice to theoretically minimize the interference between DD domain symbols. However, since a finite-duration time domain window implies infinite Doppler support, the requisite truncation compromises the theoretical orthogonality, introducing intrinsic distortion that deviates from the analytical model. Moreover, the claimed $99.99\%$ energy concentration within $[-T/2, T/2]$ for a $[-B/2, B/2]$ band-limited Doppler domain filter is contradicted by their results after performing IOTA on the truncated DD domain filters \cite{mehrotra2025pulse}.

Instead of directly modifing the DD domain filters, we propose a different method where the $\delta$-pulsone is first shaped by implementing the zeroth-order PSWF window to maximize joint time-frequency localization. Subsequently, we sample the windowed pulse set with sampling rate $F_s$ and implement the IOTA towards the sampled pulse set to enforce orthogonality on the DD domain information lattice. By formulating the samples of pulses as the matrix $\mathbf{P}$ in \eqref{Sec2_sampled_pulse_matrix} and denoting $\mathbf{R} = \mathbf{P}^{\operatorname{H}} \mathbf{P}$, the orthogonalized pulses are given by
\begin{equation}
    \tilde{\mathbf{P}} = \mathbf{P} \mathbf{R}^{-\frac{1}{2}}.
\end{equation}
It is easy to verify that
\begin{equation}
    \begin{aligned}
        \tilde{\mathbf{P}}^{\operatorname{H}} \tilde{\mathbf{P}} & = \mathbf{R}^{-\frac{1}{2}} \mathbf{P}^{\operatorname{H}} \mathbf{P} \mathbf{R}^{-\frac{1}{2}} \\
                                                                 & =\mathbf{I}_{MN}.
    \end{aligned}
\end{equation}
Then, the Zak-OTFS frame is transmitted utilizing $\tilde{\mathbf{P}}$ as the pulse shaping filter set.

\section{Numerical Results}\label{sec:numerical_results}
\begin{table}[t]
    \centering
    \caption{Power-delay profile of Veh-A channel model.}\label{tab:eva_params}
    \begin{tabular}{|c|c|c|c|c|c|c|}
        \hline Path number $(q)$                                      & $1$ & $2$    & $3$    & $4$    & $5$    & $6$    \\
        \hline $\tau_q$ ($\mu$s)                                      & $0$ & $0.31$ & $0.71$ & $1.09$ & $1.73$ & $2.51$ \\
        \hline $\frac{\mathbb{E}[|h_q|^2]}{\mathbb{E}[|h_1|^2]}$ (dB) & $0$ & $-1$   & $-9$   & $-10$  & $-15$  & $-20$  \\
        \hline
    \end{tabular}
\end{table}

In this section, we evaluate the performance of Zak-OTFS signal under designed pulse shaping filters in terms of BER, and channel estimation normalized MSE (NMSE).

\subsection{Simulation Parameters}\label{params_scheme}
\begin{table}[t]
    \centering
    \caption{Parameters setting in simulations.}\label{tab:parameters}
    \begin{tabular}{|c|c|c|}
        \hline Parameter     & Definition                & Value                       \\
        \hline $\Delta f$    & Doppler period            & $15$ (kHz)                  \\
        \hline $T$           & Delay period              & $66.67$ ($\mu$s)            \\
        \hline $M$           & Number of delay bins      & $32$                        \\
        \hline $N$           & Number of Doppler-bins    & $16$                        \\
        \hline $\beta$       & Roll-off factor           & $0.3$                       \\
        \hline $W_A$         & Frequency window duration & $(1+\beta)\cdot M \Delta f$ \\
        \hline $W_B$         & Time window duration      & $(1+\beta)\cdot N T$        \\
        \hline $f_c$         & Carrier frequency         & $4$ (GHz)                   \\
        \hline $\tau_{\max}$ & Maximum delay             & $2.51$ ($\mu$s)             \\
        \hline $\nu_{\max}$  & Maximum Doppler shift     & $815$ (Hz)                  \\
        \hline $F_s$         & Sample rate               & $10 M$                      \\
        \hline $[g_1, g_2]$  & Guard space size          & $[5, 7]$                    \\
        \hline
    \end{tabular}
\end{table}
\begin{figure*}[t]
    \centering
    \subfloat[Time domain pulse under PSWF windows.]{
        \includegraphics[width=0.315\textwidth]{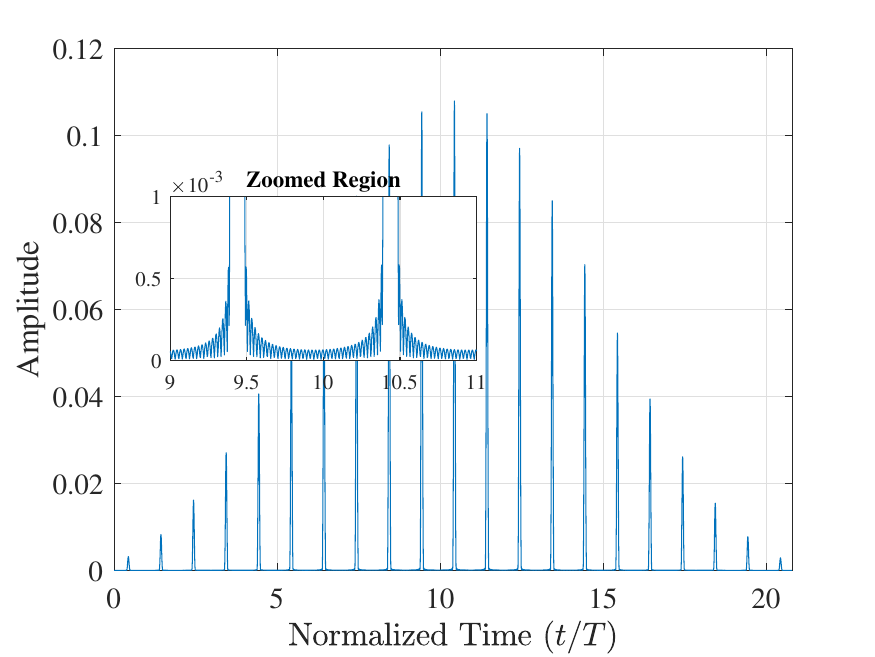}
        \label{Se4_fig_pswf_pulse}
    }
    \hspace{0.05em}
    \subfloat[Time domain pulse after performing IOTA under PSWF windows.]{
        \includegraphics[width=0.315\textwidth]{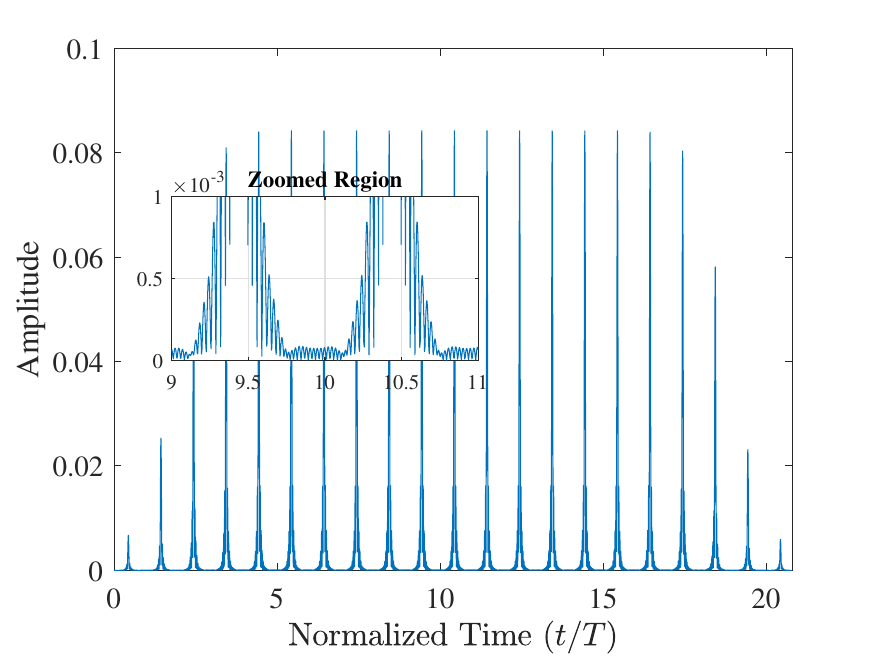}
        \label{Sec4_fig_iota_pswf_pulse}
    }
    \hspace{0.05em}
    \subfloat[Time domain pulse under RRC windows.]{
        \includegraphics[width=0.315\textwidth]{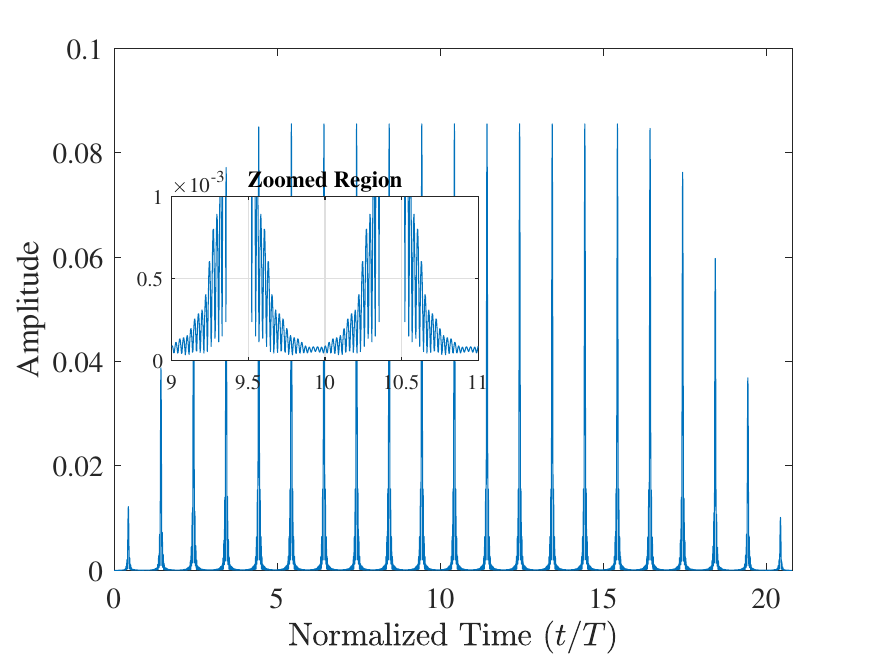}
        \label{Sec4_fig_rrc_pulse}
    }
    \caption{Visualizations of time domain pulses under different windows, $p^{\tau_{l}, \nu_{k}}_{g}(t)$ with $(\tau_{l}, \nu_{k}) = \left(\frac{15}{M \Delta f}, \frac{7}{N T}\right)$.}
    \label{Sec4_pulses_TD_visual}
\end{figure*}

We consider the six-path Vehicular A model (EVA) \cite{3gpp.36.104} with power delay profile given in Table \ref{tab:eva_params}. Unless specified otherwise, the simulation parameters are given in Table \ref{tab:parameters}. The time and frequency duration of rectangular windows are $NT$ and $M\Delta f$, respectively. The Doppler shift $\{\nu_{q}\}_{q=1}^{6}$ of the $q$-th channel path is modeled as $\nu_{\max} \cos(\theta_{q})$, where $\nu_{\max}$ is the maximum Doppler shift, and $\{\theta_{q}\}_{q=1}^{6}$ are i.i.d. random variables distributed uniformly in $[0, 2\pi)$. Based on the parameters in Table \ref{tab:parameters}, we construct the Zak-OTFS frame with the guard space as a rectangular region $\mathcal{G} = [l_p-g_1, l_p+g_1] \times [k_p-g_2, k_p+g_2]$, and the pilot region as $\mathcal{P} = [l_p-(g_1-2), l_p+(g_1-2)] \times [k_p-(g_2-2), k_p+(g_2-2)]$. The pilot is placed at $(l_p, k_p) = (M/2, N/2)$. The transmit DD domain symbols are represented as
\begin{equation}
    \hat{x}_{\operatorname{dd}}[l, k] = \begin{cases}
        \sqrt{\frac{E_d}{n_d}} \cdot x_{\mathcal{I}}[l, k], & \operatorname{if } (l, k) \operatorname{ in data region}, \\
        x_{p},                                              & \operatorname{if } (l, k) = (l_p, k_p),                   \\
        0,                                                  & \operatorname{otherwise},
    \end{cases}
\end{equation}
where $x_{\mathcal{I}}[l, k]$ are the information symbols randomly drawn from an energy-normalized constellation set (we consider 4-QAM in this paper), $x_{p} = \sqrt{E_{p}}$ is the pilot symbol, $E_{d}$ is the total energy of transmit symbols, $n_d$ is the number of transmitted data symbols, and $E_{p}$ is the energy of pilot, and the pilot-to-data power ratio (PDR) is defined in dB by
\begin{equation}
    \operatorname{PDR} = 10 \log_{10} \left(\frac{E_p}{E_d}\right),
\end{equation}
Without loss of generality, we set $E_d = n_d$. Unless specified otherwise, we set $E_p = E_d$, i.e., $\operatorname{PDR} = 0$ dB. The channel estimation is implemented according to \eqref{Sec3_channel_estimation_scheme}, and the estimation accuracy is evaluated through NMSE. We implemented $10^{4}$ Monte-Carlo simulations to obtain the average performance results, and the LMMSE detector is implemented for symbol detection at the receiver.

\subsection{Pulse Shaping Filter Visualization}

Fig. \ref{Sec4_pulses_TD_visual} illustrates the time-domain pulse shapes. The PSWF-windowed pulse, as shown in Fig. \ref{Sec4_pulses_TD_visual}\subref{Se4_fig_pswf_pulse}, exhibits the sharpest energy concentration with rapid temporal decay. After performing the IOTA transformation, the pulse shown in Fig. \ref{Sec4_pulses_TD_visual}\subref{Sec4_fig_iota_pswf_pulse} broadens the main lobe to satisfy orthogonality requirements. Then, the temporal envelope of the IOTA-PSWF pulse closely resembles that of the standard RRC pulse shown in Fig. \ref{Sec4_pulses_TD_visual}\subref{Sec4_fig_rrc_pulse}, indicating comparable temporal spreading characteristics between the two orthogonalized schemes.

\begin{figure}[t]
    \centering
    \includegraphics[width=0.99\columnwidth]{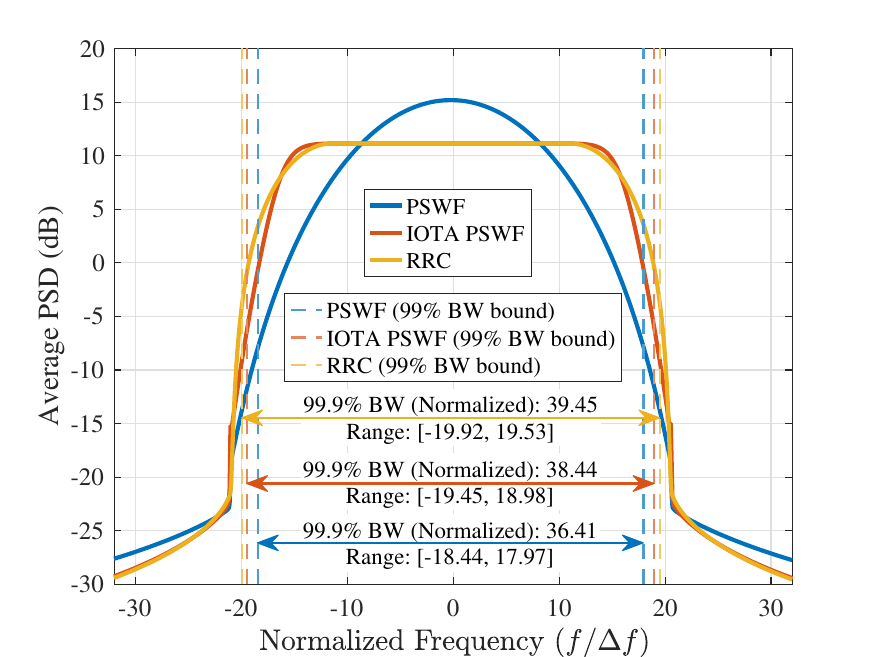}
    \caption{Average PSDs of transmit signals under different windows.}
    \label{Sec4_avg_psds}
\end{figure}

To theoretically analyze the spectral characteristics, we consider the transmitted signal as a linear combination of basis functions modulated by an independent and identically distributed (i.i.d.) symbol sequence $X_{\operatorname{dd}}[l,k]$ with zero mean and variance $\sigma_x^2$, where $\sigma_x^2 = 1$ for 4-QAM modulation. According to the Wiener-Khinchin theorem, the average Power Spectral Density (PSD), $\bar{S}(f)$, is determined by the spectral properties of the basis functions. For the proposed framework, the theoretical PSD is formulated as:
\begin{equation}
    \bar{S}(f) = \frac{\sigma_x^2}{T_{sym}} \sum_{l,k} |P^{\tau_{l}, \nu_{k}}(f)|^2,
\end{equation}
where $P^{\tau_{l}, \nu_{k}}(f)$ denotes the Fourier transform of the basis function $p^{\tau_{l}, \nu_{k}}(t)$ and $T_{sym}$ represents the symbol duration. This derivation implies that the overall frequency spectrum is fundamentally governed by the frequency response of every pulse shaping filter. Fig. \ref{Sec4_avg_psds} presents the numerical evaluation of the PSD based on this analysis. The pure PSWF scheme achieves the most localized average frequency response with a minimum bandwidth of $36.41$. After the IOTA transformation, the bandwidth expands slightly to $38.44$, and the PSD shape transitions to a flat-top profile similar to the RRC characteristic. However, the IOTA-PSWF pulse remains more spectrally efficient than the pulse under RRC window, which exhibits the largest bandwidth occupation of $39.45$.

\begin{figure}[t]
    \centering
    \includegraphics[width=0.85\columnwidth]{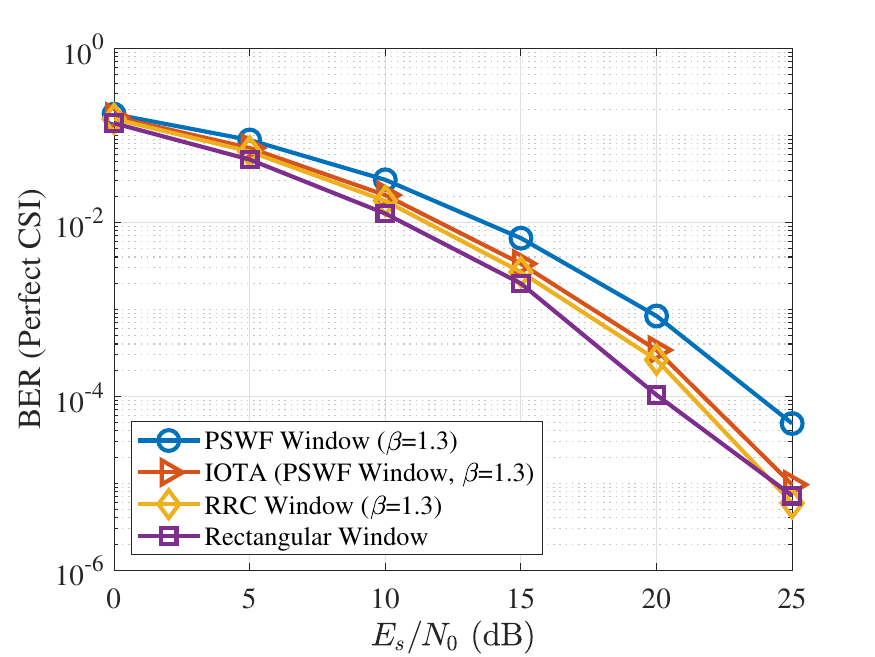}
    \caption{Uncoded 4-QAM data detection performance with perfect CSI at the receiver under different SNR.}\label{Sec4_perfectCSI_BER}
\end{figure}

\subsection{Communication Performance with Perfect CSI}
\begin{figure*}[t]
    \centering
    \subfloat[Responses of pulses under rectangular windows.]{
        \includegraphics[width=0.23\textwidth]{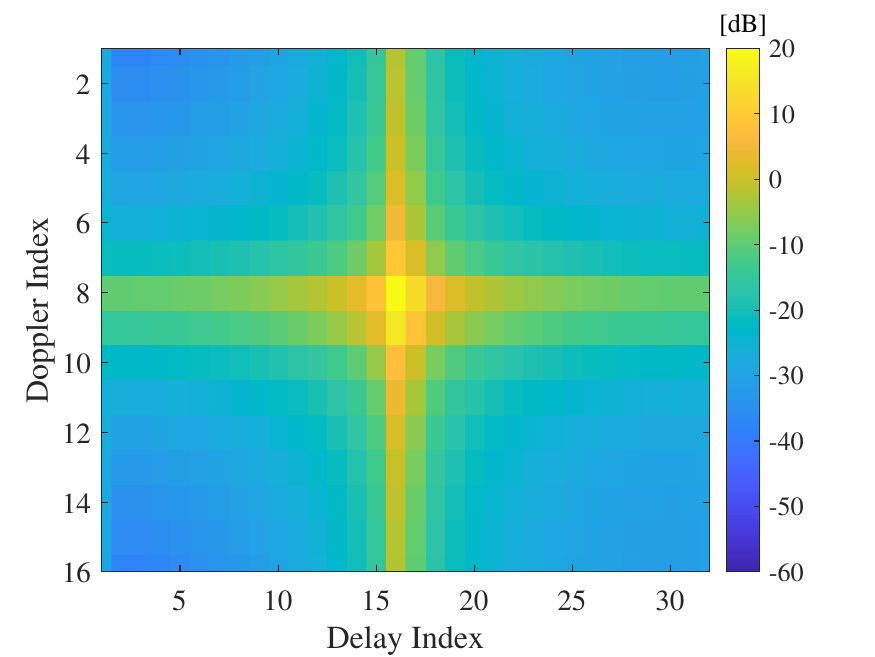}
        \label{Se4_fig_ch_eff_rect_pulse}
    }
    \hspace{0.05em}
    \subfloat[Responses of pulses under PSWF windows.]{
        \includegraphics[width=0.23\textwidth]{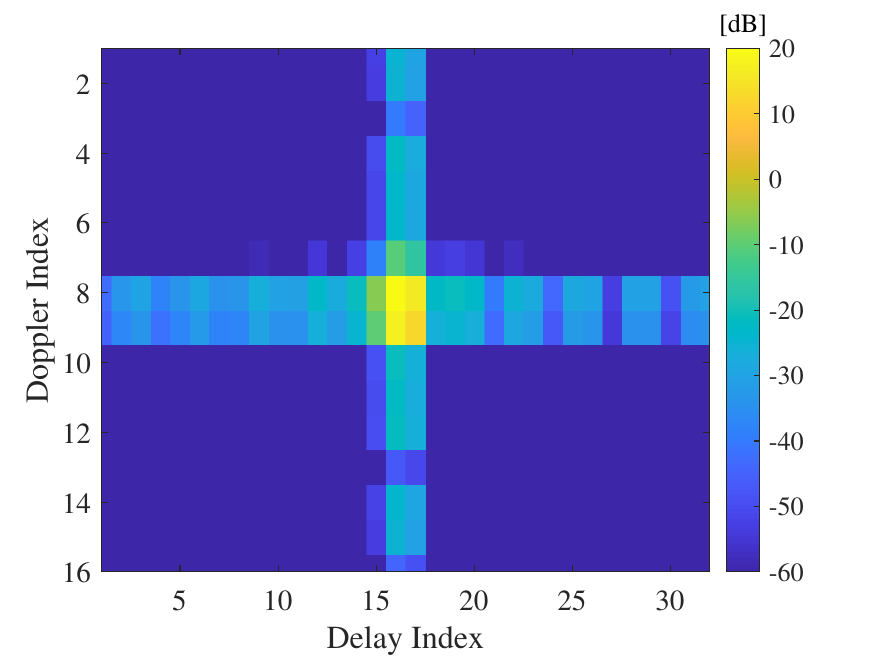}
        \label{Sec4_fig_ch_eff_pswf_pulse}
    }
    \hspace{0.05em}
    \subfloat[Responses of pulses after performing IOTA under PSWF windows.]{
        \includegraphics[width=0.23\textwidth]{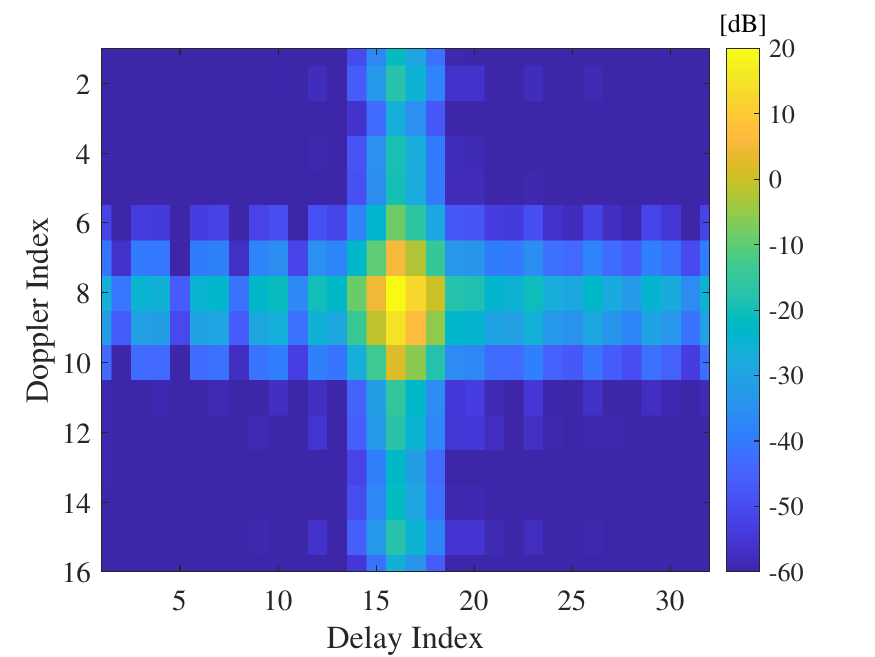}
        \label{Sec4_fig_ch_eff_iota_pswf_pulse}
    }
    \hspace{0.05em}
    \subfloat[Responses of pulses under RRC windows.]{
        \includegraphics[width=0.23\textwidth]{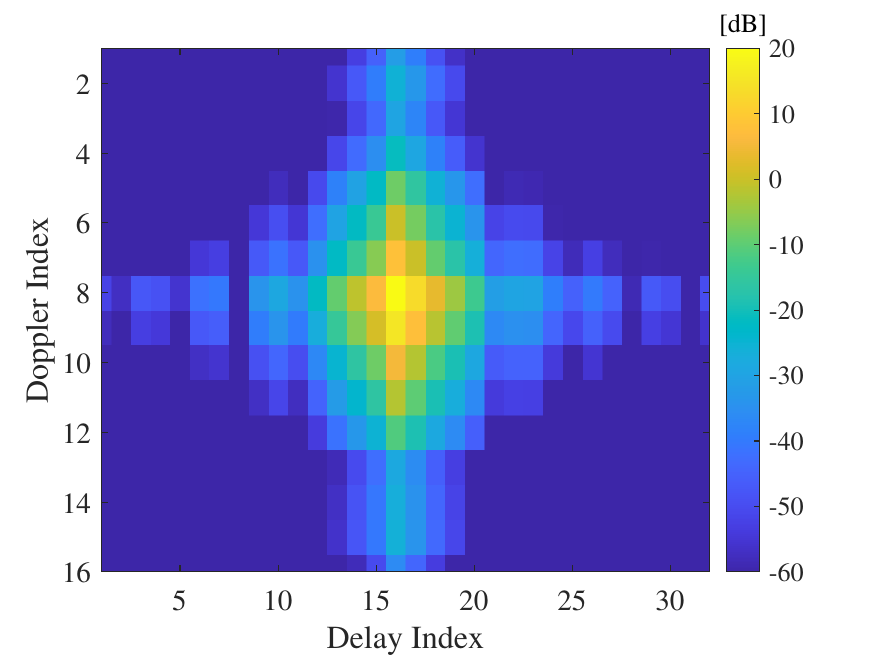}
        \label{Sec4_fig_ch_eff_rrc_pulse}
    }
    \caption{Noiseless receive responses of pilot under different windows, where the pilot is placed at $(\tau_{l}, \nu_{k}) = \left(\frac{15}{M \Delta f}, \frac{7}{N T}\right)$.}
    \label{Sec4_ch_eff_pulses}
\end{figure*}

\begin{figure*}[t]
    \centering
    \subfloat[Uncoded 4-QAM data detection performance with estimated effective channel at the receiver.\label{Sec4_estCSI_BER_GZ}]{\includegraphics[width=0.4\textwidth]{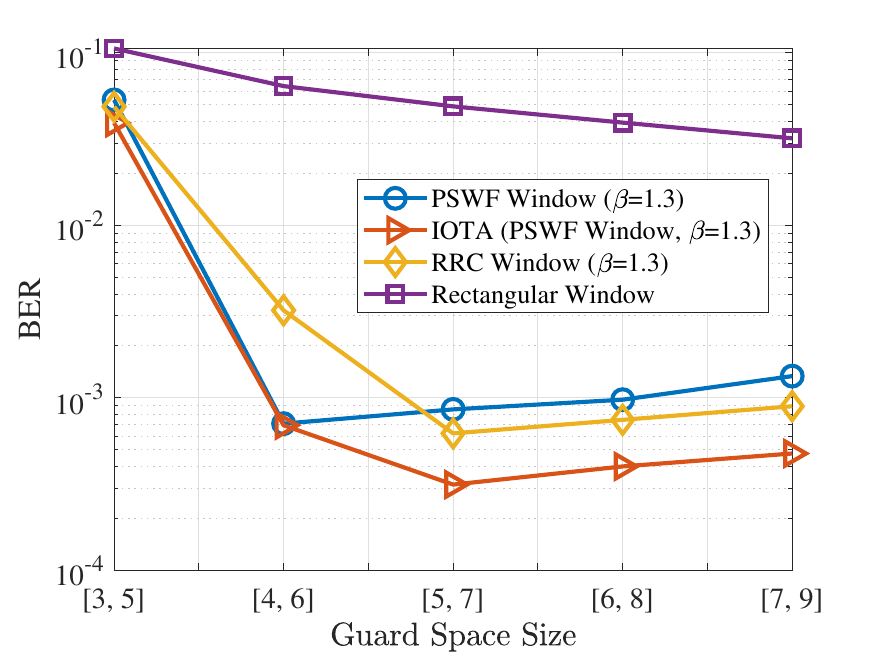}}
    \hspace{1em}
    \subfloat[Effective channel estimation NMSE.\label{Sec4_estCh_NMSE_GZ}]{\includegraphics[width=0.4\textwidth]{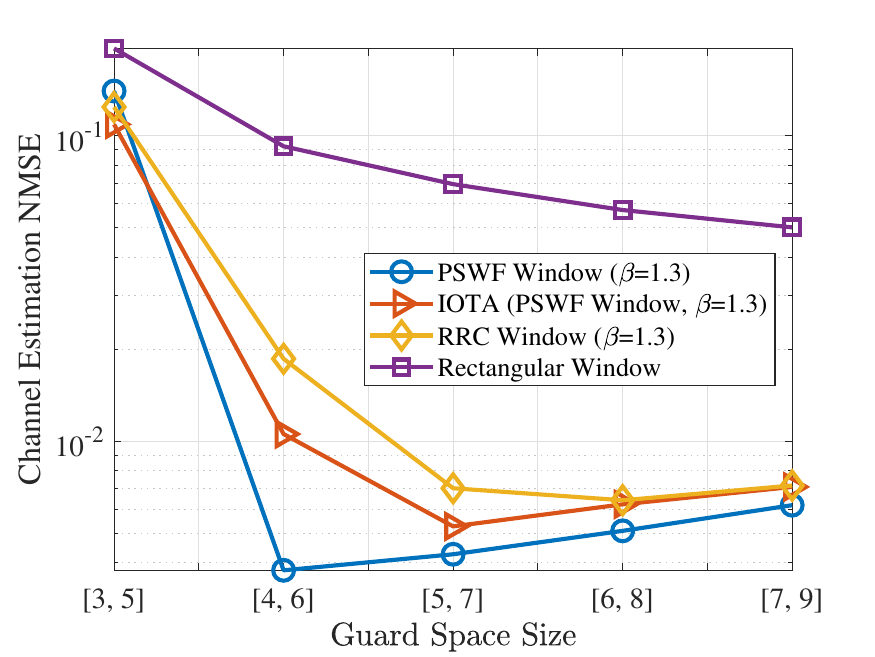}}
    \caption{Performance under different guard space sizes.}\label{Sec4_GZ_performance}
\end{figure*}
Fig. \ref{Sec4_perfectCSI_BER} evaluates the BER performance of uncoded 4-QAM signaling with perfect CSI. The BER under PSWF-windowed pulse is poor due to its non-orthogonality. The BER performance is improved after performing the IOTA procedure since it suppresses the intrinsic interference by orthogonalizing the PSWF pulse set. Furthermore, the IOTA-PSWF performance closely approaches that of the standard RRC window, with the marginal gap attributed to numerical precision limitations in computing $\mathbf{R}^{-\frac{1}{2}}$, which prevent the columns of $\tilde{\mathbf{P}}$ from satisfying strict orthogonality. The BER performance is better under the rectangular-windowed pulse set. That's because we set the power of every pulse as $1$, when the roll-off factor of the window is not $0$, the extended duration will lead to a higher energy of noise at the receiver, thereby resulting in a worse BER performance of RRC and IOTA-PSWF pulse sets.

\subsection{Channel Spreading of Different Pulses}

Fig. \ref{Sec4_ch_eff_pulses} illustrates the noiseless receive response of the pilot symbol in the DD domain after passing the same channel. The receive response under rectangular windowing, as shown in Fig. \ref{Sec4_ch_eff_pulses}\subref{Se4_fig_ch_eff_rect_pulse}, exhibits the most severe spreading, which is attributed to the high sidelobes inherent to the rectangular pulse's ambiguity function. In contrast, the receive response under the PSWF-windowed pulse, as shown in Fig. \ref{Sec4_ch_eff_pulses}\subref{Sec4_fig_ch_eff_pswf_pulse}, demonstrates the maximally localized channel spreading among these pulses. This aligns with the theoretical property of PSWF that maximize energy concentration within a prescribed time-bandwidth region, thereby yielding the lowest ambiguity sidelobes and minimal channel spreading.

After performing the IOTA procedure on the PSWF-windowed pulses, the corresponding receive response, shown in Fig. \ref{Sec4_ch_eff_pulses}\subref{Sec4_fig_ch_eff_iota_pswf_pulse}, introduces a slight expansion in channel spreading compared to the pure PSWF baseline. However, the IOTA-PSWF pulse still outperforms the RRC-windowed pulse, whose receive response is shown in Fig. \ref{Sec4_ch_eff_pulses}\subref{Sec4_fig_ch_eff_rrc_pulse}. This confirms the effectiveness of IOTA-PSWF in mitigating inter-symbol interference while preserving spectral efficiency.

\subsection{Communication Performance with Estimated CSI}

\begin{figure*}[t]
    \centering
    \subfloat[Uncoded 4-QAM data detection performance with estimated effective channel at the receiver.\label{Sec4_estCSI_BER_pilotAmp}]{\includegraphics[width=0.4\textwidth]{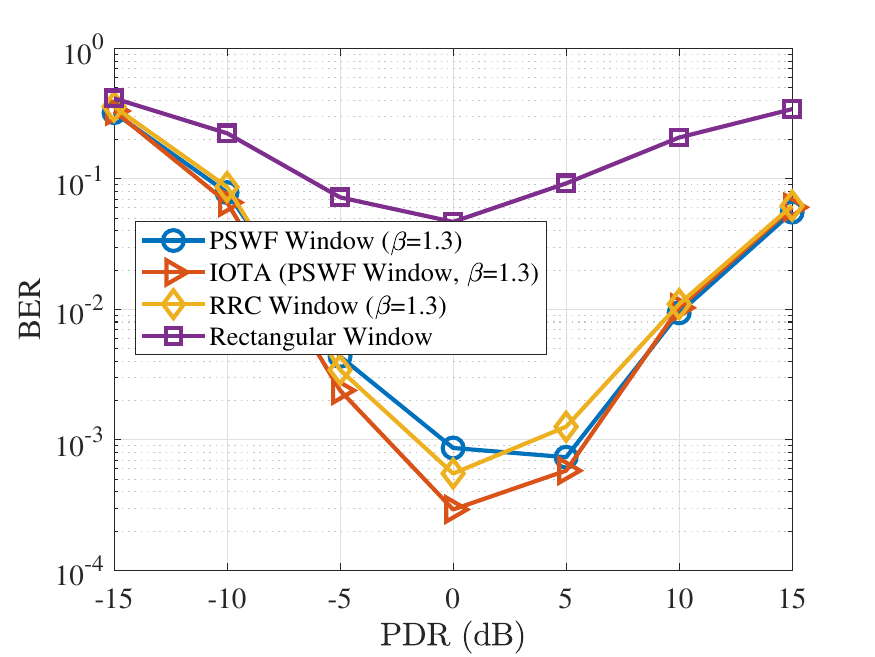}}
    \hspace{1em}
    \subfloat[Effective channel estimation NMSE.\label{Sec4_estCh_NMSE_pilotAmp}]{\includegraphics[width=0.4\textwidth]{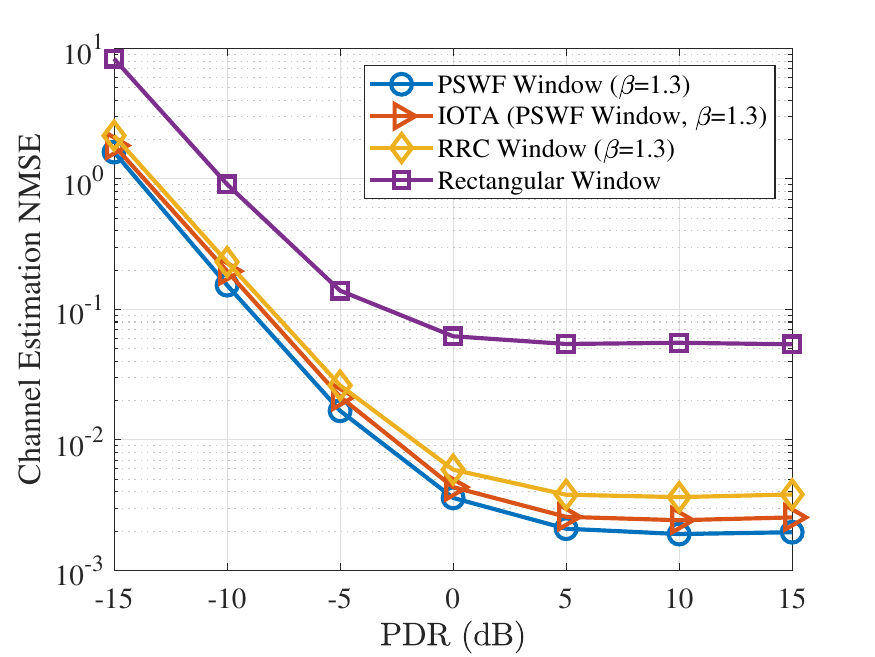}}
    \caption{Performance under different PDR.}\label{Sec4_PDR_performance}
\end{figure*}
\subsubsection{Communication Performance under Different Guard Space Sizes}

Fig. \ref{Sec4_GZ_performance} investigates the impact of guard space dimensions on system performance. In the regime of restricted guard spaces (e.g., $[3,5]$), the BER performance is poor for all these windowing schemes. This is because the small pilot region fails to sample enough channel response, leading to poor channel estimation accuracy, which further results in bad BER performance. As the guard space expands, the pilot region contains a larger proportion of the channel response, thereby yielding more precise channel estimates and enhancing symbol detection performance. However, continuously enlarging the guard space eventually degrades both the estimation accuracy and BER. That's because with a fixed pilot amplitude, an excessively large pilot region accumulates additional DD domain effective noise, thereby reducing the effective SNR for channel estimation.

Comparative analysis reveals distinct trade-offs among the windowing schemes. The pulse under PSWF window achieves the lowest NMSE due to its optimal localization; however, it exhibits poor BER performance. This discrepancy arises because the PSWF-windowed pulse set is non-orthogonal, causing severe intrinsic interference among DD domain data symbols despite the accurate channel estimates. Conversely, while the IOTA-PSWF scheme yields slightly higher channel estimation error than the PSWF-windowed pulse set, the orthogonality among its pulses eliminates intrinsic interference, thereby improving data detection.

Furthermore, the BER performance of IOTA-PSWF consistently outperforms the standard RRC window under different guard space sizes. The RRC-windowed pulse set suffers from more severe channel spreading, which generally results in worse channel estimation and detection performance. Notably, at large guard sizes (e.g., $[7,9]$), even when the IOTA-PSWF and RRC schemes yield comparable channel estimation accuracy, the IOTA-PSWF maintains a superior BER. This advantage stems from the more localized channel spreading; the broader spreading of the RRC-windowed pulse set induces more severe residual interference between data symbols, thereby degrading the BER performance.

\subsubsection{Communication Performance under Different Pilot-to-Data-Power-Ratio}
\begin{figure*}[t]
    \centering
    \subfloat[Uncoded 4-QAM data detection performance with estimated effective channel at the receiver.\label{Sec4_estCSI_BER}]{\includegraphics[width=0.4\textwidth]{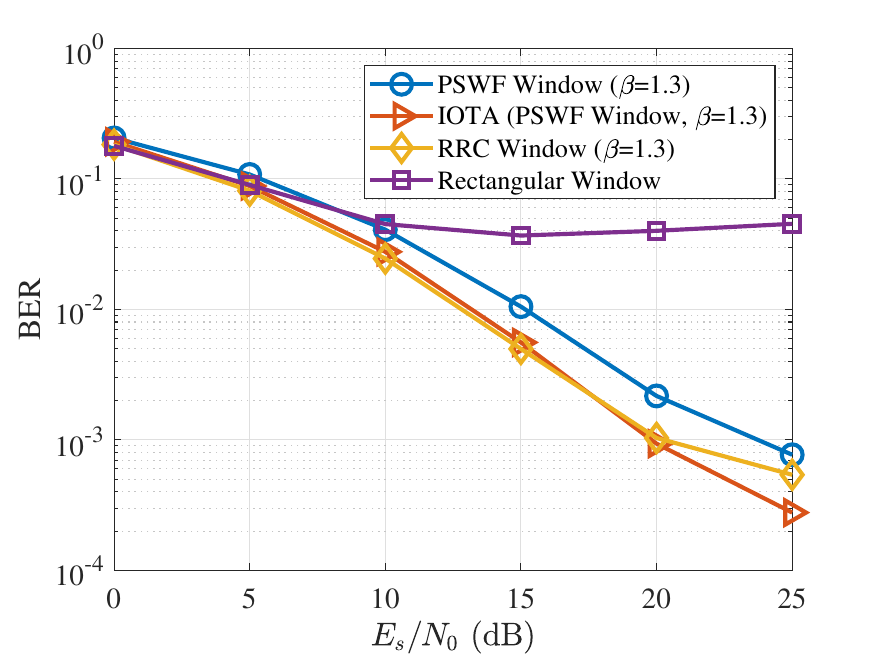}}
    \hspace{1em}
    \subfloat[Effective channel estimation NMSE.\label{Sec4_estCh_NMSE}]{\includegraphics[width=0.4\textwidth]{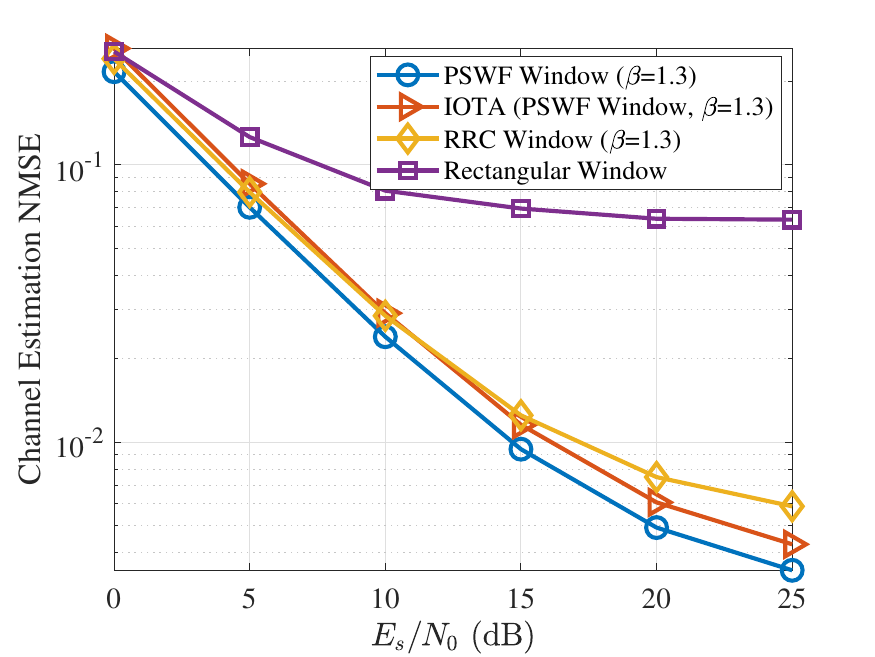}}
    \caption{Performance under different SNR.}\label{Sec4_SNR_performance}
\end{figure*}

Fig. \ref{Sec4_PDR_performance} investigates the impact of the PDR on system performance. In the low PDR regime, the BER is universally poor across all pulse shapes. This is primarily attributed to insufficient pilot energy, which yields inaccurate channel estimation that compromise data detection.

As the PDR increases, the improved channel estimation accuracy enhances the BER performance. However, beyond an optimal threshold (approximately 0 dB), the performance degrades. This reversal occurs due to the limited samples on channel response; increasing the pilot power cannot compensate for the information loss caused by the finite sampling (i.e., missing channel information outside the pilot region). Crucially, as the PDR continues to rise, these unsampled components of the high-power pilot response imposing more severe interference on the data symbols and deteriorating the BER. Among these pulse sets, the proposed IOTA-PSWF pulse set consistently achieves the best trade-off.

\subsubsection{Communication Performance under Different SNR}

Fig. \ref{Sec4_SNR_performance} evaluates the system performance across varying SNRs. In the low SNR regime, both channel estimation accuracy and BER are poor due to noise dominance. As SNR increases, the channel estimation precision improves for all pulse shapes, leading to corresponding gains in BER performance.

The rectangular window exhibits the poorest performance. Its high ambiguity function sidelobes cause severe channel spreading, meaning the receiver samples a smaller fraction of the total channel response under the same size of pilot region, resulting in inaccurate channel estimation and degraded detection. Conversely, the pure PSWF window achieves the lowest NMSE due to its optimal localization. However, its BER performance is compromised by the lack of orthogonality among pulses, which induces interference in the DD domain symbols.

The proposed IOTA-PSWF pulse set effectively addresses this limitation. By restoring orthogonality, its BER significantly outperforms the pulse set under the PSWF windowing while maintaining comparable estimation accuracy. At low SNRs, the IOTA-PSWF and RRC-windowed pulse sets show similar performance despite the former's more localized channel response. This is because, in low SNR regimes, the power of noise overshadowing the benefits of more localized effective channel response.

\section{Conclusion}\label{sec:conclusion}

This paper presented a comprehensive framework for pulse shaping in Zak-OTFS systems, moving from theoretical derivation to practical design. By establishing the discrete-time oversampled I/O relationship, we proved that the effective channel response remains consistent across all DD domain symbols. Leveraging this model, we identified the ambiguity function sidelobes as the key factor degrading channel estimation and proposed the IOTA-PSWF design to address this issue. Extensive simulations confirm that our proposed scheme effectively mitigates channel spreading and achieves superior error performance compared to standard windowing techniques. These findings provide a solid foundation for the practical deployment of Zak-OTFS in high-mobility scenarios.
\begin{figure*}[t]
    \begin{multline}\label{apdx_q_th_eff_matrix_kl}
        H_{\operatorname{eff}, q}[k^{\prime}M + l^{\prime}, kM + l] = h_{q} (\mathbf{p}_{g}^{\tau_{l^{\prime}}, \nu_{k^{\prime}}})^{\operatorname{H}} \mathbf{J}_{l_{\tau_{q}}} \mathbf{D}_{k_{\nu_{q}}} \mathbf{p}_{g}^{\tau_{l}, \nu_{k}} = \sum_{n = 0}^{LMN - 1} h_{q} (p_{g}^{\tau_{l^{\prime}}, \nu_{k^{\prime}}}[(n+l_{\tau_{q}})_{LMN}])^{*} p_{g}^{\tau_{l}, \nu_{k}}[n] e^{j 2 \pi \frac{n k_{\nu_{q}}}{LMN}}\\
        = h_{q} \sum_{n = 0}^{LMN - 1} B[(n + l_{\tau_{q}})_{LMN}] \sum_{m^{\prime} \in \mathbb{Z}} A(m^{\prime} \Delta f + \nu_{k^{\prime}}) e^{-j 2 \pi (m^{\prime} N + k^{\prime}) \frac{n + l_{\tau_{q}} - L l^{\prime}}{LMN}} B[n] \\
        \times \sum_{m \in \mathbb{Z}} A(m \Delta f + \nu_{k}) e^{j 2 \pi (m N + k) \frac{n - L l}{LMN}} e^{j 2 \pi \frac{n k_{\nu_{q}}}{LMN}}\\
        = h_{q} e^{-j 2 \pi \frac{l_{\tau_{q}} k_{\nu_{q}}}{LMN}} \hspace{-1em}\sum_{(m, m^{\prime}) \in \mathbb{Z} \times \mathbb{Z}} \mathcal{X}_{\tilde{B}}(-l_{\tau_{q}}, k^{\prime} - k - (m - m^{\prime}) N - k_{\nu_{q}}) A(m^{\prime} \Delta f + \nu_{k^{\prime}}) A(m \Delta f + \nu_{k})\\
        \times e^{j 2 \pi (m^{\prime} N + k^{\prime}) \frac{l^{\prime}}{MN}} e^{-j 2 \pi (m N + k) \frac{l + l_{\tau_{q}}}{MN}}
    \end{multline}\normalsize
\end{figure*}
\begin{figure*}[t]
    \begin{multline}\label{apdx_AF_freq_window}
        \sum_{m \in \mathbb{Z}} e^{j 2 \pi \frac{k m}{N}} \mathcal{Y}_{A}(\tau_{l^{\prime}} - \tau_{l} - \tau_{q} - m T, -(\nu_{k^{\prime}} - \nu_{k} - \bar{m} \Delta f)) \\
        = \int A(f) A(f + \nu_{k^{\prime}} - \nu_{k} - \bar{m} \Delta f) e^{j 2 \pi f (\tau_{l^{\prime}} - \tau_{l} - \tau_{q})} \sum_{m \in \mathbb{Z}} e^{-j 2 \pi m T (f - \frac{k}{NT})} \, df,
    \end{multline}
\end{figure*}
\begin{figure*}[t]
    \begin{multline}\label{apdx_AF_frew_window_3}
        \sum_{m \in \mathbb{Z}} A(m^{\prime} \Delta f + \nu_{k^{\prime}}) A(m \Delta f + \nu_{k}) e^{j 2 \pi (m^{\prime} N + k^{\prime}) \frac{l^{\prime}}{MN}} e^{-j 2 \pi (m N + k) \frac{l + l_{\tau_{q}}}{MN}}\\
        = \sum_{m \in \mathbb{Z}} e^{j 2 \pi \frac{k m}{N}} \mathcal{Y}_{A}(\tau_{l^{\prime}} - \tau_{l} - \tau_{q} - m T, -(\nu_{k^{\prime}} - \nu_{k} - \bar{m} \Delta f)) e^{j 2 \pi \frac{(k^{\prime} - k - \bar{m} N) (l^{\prime} - l - m N)}{MN}} e^{j 2 \pi \frac{(k^{\prime} - k - \bar{m} N) (l + m N)}{MN}}
    \end{multline}
\end{figure*}
{
\appendices
\section{Proof of Proposition \ref{Sec2_proposition_effective_channel_expression}}\label{apdx_effective_channel}
The $(k^{\prime}M + l^{\prime}, kM + l)$-th element in $\mathbf{H}_{\operatorname{eff}}$ of \eqref{Sec2_discrete_H_matrix} at the $q$-th path is given in \eqref{apdx_q_th_eff_matrix_kl}, where $\mathcal{X}_{\tilde{B}}(l, k)$ is the discrete periodic ambiguity function defined in \eqref{Sec2_discrete_periodic_AF_B}. Since we have the equation in \eqref{apdx_AF_freq_window}, by using the result of Poisson summation, i.e.,
\begin{equation}
    \sum_{m \in \mathbb{Z}} e^{-j 2 \pi m T (f - \frac{k}{NT})} = T \sum_{m \in \mathbb{Z}} \delta(f - \frac{k}{NT} - m \Delta f),
\end{equation}
and denoting $\bar{m} = m - m^{\prime}$, equation \eqref{apdx_AF_freq_window} becomes
\begin{multline}
    \sum_{m \in \mathbb{Z}} e^{j 2 \pi \frac{k m}{N}} \mathcal{Y}_{A}(\tau_{l^{\prime}} - \tau_{l} - \tau_{q} - m T, -(\nu_{k^{\prime}} - \nu_{k} - \bar{m} \Delta f))\\
    =\sum_{m \in \mathbb{Z}} A(\nu_{k} + m \Delta f) A^{*}((m - \bar{m}) \Delta f + \nu_{k^{\prime}}) e^{j 2 \pi \frac{(l^{\prime} - l - l_{\tau_{q}}) (m N + k)}{M N}},\nonumber
\end{multline}
which further gives the result in \eqref{apdx_AF_frew_window_3}. By substituting \eqref{apdx_AF_frew_window_3} into \eqref{apdx_q_th_eff_matrix_kl}, and summing $q$ form $1$ to $Q$, it is clear that Proposition \ref{Sec2_proposition_effective_channel_expression} holds.
}

\bibliographystyle{IEEEtran}
\bibliography{refs}
\vfill
\end{document}